\documentclass[lettersize,journal]{IEEEtran}
\usepackage{cite}
\usepackage{amsmath,amssymb,amsfonts}
\usepackage{algorithmic}
\usepackage{graphicx}
\usepackage{algorithm,algorithmic}
\newtheorem{theorem}{Theorem}
\newtheorem{corollary}{Corollary}
\newtheorem{definition}{Definition}

\newtheorem{remark}{Remark}

\newtheorem{lemma}{Lemma}

\usepackage{hyperref}
\usepackage{textcomp}
\def\BibTeX{{\rm B\kern-.05em{\sc i\kern-.025em b}\kern-.08em
    T\kern-.1667em\lower.7ex\hbox{E}\kern-.125emX}}
\begin{document}
\title{Robust Semi-passive Velocity Field Control with Boundedness Guarantees for Safe Interaction between Mechanical Systems and Physical Environment}
\author{Van Trong Dang, Sumitaka Honji, and Takahiro Wada
\thanks{This work was partially supported by JSPS KAKENHI (Grant Number 24H00298), Japan.}
\thanks{The authors are with the Graduate School of Science and Technology, Nara Institute of Science and Technology, Ikoma, Nara 630-0192, Japan (e-mail: van\_trong.dang.ve6@naist.ac.jp, honji.sumitaka@naist.ac.jp, and t.wada@is.naist.jp).}
%\thanks{Second B. Author Jr. was with Rice University, Houston, TX 77005 USA. He is  now with the Department of Physics, Colorado State University, Fort Collins,  CO 80523 USA (e-mail: author@lamar.colostate.edu).}
%\thanks{Third C. Author is with  the Electrical Engineering Department, University of Colorado, Boulder, CO  80309 USA, on leave from the National Research Institute for Metals, Tsukuba, Japan (e-mail: author@nrim.go.jp).}
}

\maketitle

\begin{abstract}
Controllers that guarantee energetic passivity with respect to the pair of external force and velocity realize safe interaction between the mechanical system and its physical environment. However, solely adhering to energetic passivity constraints may impose fundamental limitations on control performance and, in some cases, prevent the successful execution of controlled tasks. In addition, external disturbances from the physical environment can drive the system’s energy level and states beyond operational regions, thereby undermining task performance and safety. In this paper, we study a robust time-varying semi-passive velocity field control to aim to relax the inherently conservative nature of fully passive control methods in a controlled manner. Specifically, the proposed control method guarantees passivity of the closed-loop system with respect to the force–velocity
input–output pair when the energy level exceeds a predefined level, while permitting non-passive behaviors to preserve task performance otherwise. Furthermore, the energy level and the states of the closed-loop system are proved to converge to bounded domains even in the presence of unpredicted disturbances, respectively. Additionally, the proposed method also enables constraining power flow between the closed-loop system and its physical environment to enhance safety in the interaction process. Numerical simulation examples demonstrate the effectiveness of the proposed method.
\end{abstract}

\begin{IEEEkeywords}
Lyapunov method, robust control, semi-passive velocity field control (SPVFC), uniform ultimate boundedness (UUB).
\end{IEEEkeywords}

\section{Introduction}
\label{sec:introduction}
\IEEEPARstart{E}{nsuring} stability and safety during physical interaction with unstructured environments plays a central role in controlling the class of nonlinear mechanical systems, such as robot manipulators, mobile robots, inverted pendulum systems, and so on \cite{c1,c2}. In this context, the concept of energetic passivity in \cite{c3,c4,c5} with respect to the force–velocity input–output pair provides a principled theory for safe interaction. Specifically, the concept enables the closed-loop system to interact in an energetically passive manner with its physical environment, rather than primarily minimizing tracking error as sliding mode control in \cite{c6} or constraining state as barrier Lyapunov function approach in \cite{c7} without considering energy relations. Although the energetic passivity of the feedback interconnection between the mechanical system and its physical environment provides a sufficient condition for input–output stability and is widely used to promote safe interaction \cite{c8}, its inherent conservatism may significantly restrict the achievable performance of control systems during physical interaction tasks. For example, a passive system may exhibit sluggish behaviors or even stalled responses because it solely dissipates energy supplied by the external environment without generating energy internally. Meanwhile, upper-bound energy constraints are insufficient to ensure safety in real-world applications due to ignoring the limitation of the energy transfer rate (i.e., power flow). Despite its passive nature, the system may exhibit significant energy and instantaneous power flow that give rise to undesired or dangerous behaviors in unstructured environments. Therefore, it is necessary to evolve a new passivity-based control approach satisfying stability and safe physical interaction while relaxing the conservative nature in classical passive formulations.

In recent decades, passivity-based control (PBC) approaches have been extensively developed and investigated in the literature. Accordingly, various notable works \cite{c12,c13,c14,c15,c16,c17,c18,c19,c20,c21,c22,c23,c24} have emerged that seek to generalize or maintain passivity at all time. Specifically, taking the benefits of sliding mode control, robust PBC methods in \cite{c12,c13,c14} were applied to nonlinear mechanical systems under external perturbations.  To satisfy both passive and safety simultaneously, the ideas of PBC and control barrier functions were combined in the studies \cite{c15,c16}. In this manner, an enhanced method in \cite{c17} rendered closed-loop passivity and a safety-critical set forward invariant, even in the presence of external disturbances. Furthermore, in the context of physical interaction with unmodeled dynamic environments, the concept of the energy tank in \cite{c18,c19,c20} is widely used by leveraging energy accumulation and redistribution mechanisms. To enhance overall safety, an adaptive PBC using a virtual energy tank was proposed in \cite{c21}, enabling the limitation of power flow. Moreover, passive velocity field control (PVFC) approaches in \cite{c22,c23,c24} were designed for contour following tasks by combining the concepts of the fictitious flywheel serving as a fictitious energy storage component and time-invariant velocity field. However, a drawback of the aforementioned PBC approaches is their inherently conservative nature when enforcing tight passivity conditions. This is because the task performance can significantly degrade in the case of energy-depleted conditions. Consequently, while satisfying passivity conditions, passive systems may be unable to fully execute the intended tasks. 

%Recently, extensions in \cite{c9,c10,c11} of classical passivity theory have been proposed to address its limitations in sampled-data and discrete-time domains. Specifically, Krasovskii passivity in \cite{c9,c10} was introduced to sampled-data stabilization and output consensus control, while a new notion of $\varrho$-passivity in \cite{c11} provides a unified framework to preserve passivity across abstract time domains. 

For relaxing the conservative nature of PBC, it is necessary to incorporate controlled energy-compensation mechanisms that enable more flexible and responsive system behavior while preserving overall stability. Inspired by the original PVFC \cite{c22,c23,c24}, the works in \cite{c25,c26} introduced velocity field control with energy compensation that adaptively injects or dissipates energy based on a predefined threshold. Accordingly, an improved velocity field control was developed in \cite{c27} to compensate energy within a finite-time interval while ensuring the stability of the closed-loop system. Furthermore, for balancing performance and stability in physical interaction, an ultimate passivity approach in \cite{c28} was proposed based on the switching principle between nominal and conservative control modes. Although these studies mitigated the inherent conservativeness of PBC by introducing controlled energy supplementation, the occurrence of high energy and instantaneous power flow may compromise safety during physical interaction.

Motivated by the aforementioned observations, this paper develops a robust time-varying semi-passive velocity field control (SPVFC) for the class of nonlinear mechanical systems due to its representativeness and generality. Specifically, the proposed semi-passive control method guarantees that the closed-loop mechanical system is passive with respect to the force–velocity input–output pair when the energy level exceeds a predefined level, while allowing non-passive behaviors to preserve task performance otherwise. The robust SPVFC approach facilitates a systematic relaxation of the inherent conservativeness of PBC \cite{c18,c19,c20,c21,c22,c23,c24} to ensure task performance while simultaneously enhancing safety during physical interaction with unstructured environments. In contrast to the switching-based passivity approach in \cite{c28}, which may induce discontinuities and lead to undesirable oscillatory behavior, the proposed method ensures a smooth and continuous transition between passive and non-passive regimes. In addition, the present study extends and addresses the remaining problems of the preliminary work presented in \cite{c27}, which did not account for inherent external disturbances from the physical environment and power flow constraints. In the present study, it is worth emphasizing that the energy-compensation mechanism ensures the system's energy level converges rapidly to a designed region within a finite-time interval. Furthermore, the states of the closed-loop mechanical system are also proven to be uniformly ultimately bounded (UUB) with an exponential convergence rate despite the existence of external disturbances. In addition to relaxing the conservative nature and achieving the desired task, the power flow under the proposed approach is bounded for enhancing safety during interaction. By appropriately configuring the control parameters, the operator can define the bounded operating regions, enforce the safety-invariant set, and regulate the convergence rate to meet the requirements of specific applications.

The remainder of this paper is organized as follows. Section II gives the mathematical preliminaries and foundational theories. In Section III, we introduce the robust time-varying semi-passive velocity field control method and the main contributions. The effectiveness of the proposed approach is investigated through numerical simulation examples in Section IV. Section V presents some discussions and limitations of the proposed method. Finally, Section VI concludes this paper.

\textbf{\textit{Notations:}} $\mathbb{R}$ and $ {{\mathbb{R}}_{+}}$ are the sets of real numbers and non-negative real numbers, respectively. The $n\times n$ identity matrix is denoted by ${{\mathbf{I}}_{n}}$. For a vector $\mathbf{x}=\left[ \begin{matrix}
   {{x}_{1}} & {{x}_{2}} & \ldots  & {{x}_{n}}  \\
\end{matrix} \right]^{\top}\in {{\mathbb{R}}^{n}}$, $\text{diag}\left( \mathbf{x} \right)$ is a diagonal matrix constructed by $\mathbf{x}$. The symbols ${{\left\| \mathbf{x} \right\|}_{p}}$ and ${{\left\| \mathbf{x} \right\|}_{\infty}}$ are the
$p-$norm and infinity norm of $\mathbf{x}$, respectively. These are defined by ${{\left\| \mathbf{x} \right\|}_{p}}:={{\left( {{\left| {{x}_{1}} \right|}^{p}}+{{\left| {{x}_{2}} \right|}^{p}}+\ldots +{{\left| {{x}_{n}} \right|}^{p}} \right)}^{\frac{1}{p}}}$ and ${{\left\| \mathbf{x} \right\|}_{\infty }}:=\underset{1\le i\le n}{\mathop{\max }}\,\left| {{x}_{i}} \right|$. In particular, the notation $\left\| \mathbf{x} \right\|$ is used to represent the Euclidean norm ${{\left\| \mathbf{x} \right\|}_{2}}$. The generalized power of real numbers is defined as ${{\left\lfloor{\mathbf{x}}\right\rceil}^{\zeta}}:={{\left[ {\mathop{\rm \text{sgn}}} ({x_1}){\left| {{x_1}} \right|^\zeta},{\mathop{\rm \text{sgn}}} ({x_2}){\left| {{x_2}} \right|^\zeta},\ldots ,{\mathop{\rm \text{sgn}}} ({x_n}){\left| {{x_n}} \right|^\zeta} \right]}^{\top}}$ with $\zeta \in {{\mathbb{R}}_{+}}$. The gradient of the time-varying function is $ {{\dot{f}}}\left( \mathbf{x},t \right)=\sum\limits_{i=1}^{n}{\frac{\partial {{f}}\left( \mathbf{x},t \right)}{\partial {{x}_{i}}}{{{\dot{x}}}_{i}}}+\frac{\partial {{f}}\left( \mathbf{x},t \right)}{\partial t}$. 
%The symbol ${{\nabla }_{\left( \centerdot  \right)}}\mathbf{f}\left( \mathbf{x},t \right)$ denotes the gradient  of the time-varying vector function  $\mathbf{f}\left( \mathbf{x},t \right)={{\left[ \begin{matrix}
 %  {{f}_{1}}\left( \mathbf{x},t \right) & {{f}_{2}}\left( \mathbf{x},t \right) & \ldots  & {{f}_{m}}\left( \mathbf{x},t \right) \\
%\end{matrix} \right]}^{\top}}\in {{\mathbb{R}}^{m}}$ with respect to $\left( \centerdot  \right)$, that is, ${{\nabla }_{t}}\mathbf{f}\left( \mathbf{x},t \right):={{\left[ \begin{matrix}
  % {{\partial }_{t}}{{f}_{1}}\left( \mathbf{x},t \right) & {{\partial }_{t}}{{f}_{2}}\left( \mathbf{x},t \right) & \ldots  & {{\partial }_{t}}{{f}_{m}}\left( \mathbf{x},t \right)  \\
%\end{matrix} \right]}^{\top}}$ 
%in which  ${{\partial }_{t}} {{f}_{j}}\left( \mathbf{x},t \right):=\sum\limits_{i=1}^{n}{\frac{\partial {{f}_{j}}\left( \mathbf{x},t \right)}{\partial {{x}_{i}}}{{{\dot{x}}}_{i}}}+\frac{\partial {{f}_{j}}\left( \mathbf{x},t \right)}{\partial t}$, $j=1,2,\ldots m$.

\section{Preliminaries}
\label{sec2}

\subsection{Problem Formulation}
Considering the multi-input multi-output (MIMO) nonlinear mechanical system described by:
\begin{equation}
\label{eq1}
\mathbf{M}\left( \mathbf{q} \right)\mathbf{\ddot{q}}+\mathbf{C}\left( \mathbf{q},\mathbf{\dot{q}} \right)\mathbf{\dot{q}}=\pmb{\tau }+{{\pmb{\tau }}_{ext}},
\end{equation}
in which $\mathbf{q}={{\left[ \begin{matrix}
   {{q}_{1}} & {{q}_{2}} & \ldots  & {{q}_{n}}  \\
\end{matrix} \right]}^{\top}},\,\mathbf{\dot{q}},\,\mathbf{\ddot{q}}\in {{\mathbb{R}}^{n}}$ are the position, velocity, and acceleration vectors of the system, respectively. $\mathbf{M}\left( \mathbf{q} \right)\in {{\mathbb{R}}^{n\times n}}$ denotes the inertia matrix which is symmetric and positive definite and $\mathbf{C}\left( \mathbf{q},\mathbf{\dot{q}} \right)\in {{\mathbb{R}}^{n\times n}}$ is Coriolis and centrifugal matrix. $\left( {{{\mathbf{\dot{M}}}}}\left( {{\mathbf{q}}} \right)-2{{\mathbf{C}}}\left( {{\mathbf{q}}},{{{\mathbf{\dot{q}}}}} \right) \right)$ is a skew-symmetric matrix. $\pmb{\tau }$ and ${{\pmb{\tau }}_{ext}}\in {{\mathbb{R}}^{n}}$ are the generalized control inputs and external torque from the physical environment, respectively. $\mathbf{y}={{\left[ \begin{matrix}
   {{\mathbf{q}}^{\top}} & {{{\mathbf{\dot{q}}}}^{\top}}  \\
\end{matrix} \right]}^{\top}}\in {{\mathbb{R}}^{2n}}$ denotes the system output vector. In this study, we assume that the external torque vector is bounded. This is because the external torque from the physical environment typically has finite energy, thus rendering the reasonableness of the above assumption. In addition, as the gravity term can be straightforwardly compensated using established techniques (e.g., as demonstrated in Chapter 2 of \cite{c29}), we exclude this term from the dynamics model to mitigate the complexity involved in designing and analyzing the controller in the following sections.

For physical interactions between mechanical systems with unstructured environments, a fundamental challenge lies in balancing task performance and interaction safety. First, during physical interactions, the mechanical system is subjected to external disturbances from the environment, which pose significant challenges in regulating both motion and energy dynamics. For instance, such disturbances can lead to an excessive decrease or increase in the system's energy level, thereby causing undesirably sluggish or overly rapid responses, respectively. This implies that the task performance or interaction safety can be compromised due to external disturbances. Additionally, even when the closed-loop mechanical system satisfies the upper-bound energy constraints, the instantaneous power flow between the system and environment may still induce undesirable or potentially hazardous behaviors. Second, while purely passive behavior may cause the system to stall or fail to complete tasks, overly aggressive actions can compromise safety and violate the stability requirements of physical interaction. Thus, the development of a robust time-varying SPVFC for the closed-loop mechanical system, which is capable of relaxing conservative nature, maintaining bounded energy and motion, and enforcing power flow constraints, is expected to ensure safety while achieving the desired task performance, even in the presence of external disturbances.

%For physical interactions between mechanical systems with unstructured environments, a fundamental challenge lies in balancing task performance and interaction safety. Indeed, exhibiting purely passive behavior may lead the system to stall or fail to accomplish tasks, whereas overly aggressive actions can compromise safety and violate interaction stability requirements. When the closed-loop mechanical system is regulated by the concept of SPVFC, passivity can be relaxed in a controlled manner to maintain task performance. Specifically, SPVFC enables the system to temporarily generate energy internally, within a well-defined energy-bounded region, while ensuring that such non-passive behavior is strictly confined and recoverable. In addition, the amount of energy and the power flow that can injure the system and environment are constrained, thereby enhancing safety during physical contacts. Thus, by applying the concept of robust time-varying SPVFC, the closed-loop mechanical system is expected to achieve the desired task performance while simultaneously ensuring safety, even in the presence of uncertain nature.

\subsection{Fundamental Theories}

\begin{definition} 
\label{def1} {\emph{(Passivity)}}
\cite{c2} A dynamic system with input $\mathbf{u}$, output $\mathbf{y}$, and state $\mathbf{{z}}$ is passive if there exist a non-negative storage function $S\left( \mathbf{z} \right)$ and a non-negative dissipation function $D\left( \mathbf{z} \right)$ such that the following is satisfied for any $t\in {{\mathbb{R}}_{+}}$: 
\begin{equation}
\label{eq2}
\int\limits_{0}^{t}{{{\mathbf{y}}^{\top}}\left( \varepsilon  \right)\mathbf{u}\left( \varepsilon  \right)}d\varepsilon -\int\limits_{0}^{t}{D\left( \mathbf{z }\left( \varepsilon  \right) \right)}d\varepsilon \ge -S\left( \mathbf{z }\left( 0 \right) \right).
\end{equation}
in which the amount of energy is dissipated that is given by $\int\limits_{0}^{t}{D\left( \mathbf{z }\left( \varepsilon  \right) \right)}d\varepsilon$.
\end{definition}

\begin{theorem}
\label{thm1} {\emph{(UUB by Lyapunov Analysis)}} \cite{c30,c31}
For a nonlinear system with an unknown but bounded disturbance, there exists a continuously partially differentiable function $L\left( \mathbf{z},t \right):\left( \mathbb{D}\subset {{\mathbb{R}}^{n}} \right)\times {{\mathbb{R}}_{+}}\to \mathbb{R}$ and class $\mathcal{K}$ functions ${{\alpha }_{1}}\left( \centerdot  \right)$, ${{\alpha }_{2}}\left( \centerdot  \right)$, ${{\alpha }_{3}}\left( \centerdot  \right)$ such that: 
\begin{equation}
\left\{ \begin{aligned}
  & {{\alpha }_{1}}\left( \left\| \mathbf{z} \right\| \right)\le L\left( \mathbf{z},t \right)\le {{\alpha }_{2}}\left( \left\| \mathbf{z} \right\| \right),\,\,\,\,\left( \mathbf{z},t \right)\in \mathbb{D}\times {{\mathbb{R}}_{+}} \\ 
 & {{\partial }_{t}}L\left( \mathbf{z},t \right)\le -{{\alpha }_{3}}\left( \left\| \mathbf{z} \right\| \right),\,\,\,\,\left( \mathbf{z},t \right)\in \mathbb{D}\times {{\mathbb{R}}_{+}},\,\left\| \mathbf{z} \right\|>B \\ 
\end{aligned} \right.
\label{eq3}
\end{equation}
for some $B>0$ where the ball of radius $B$  is contained in $\mathbb{D}$, then the system is UUB with bound $\alpha _{1}^{-1}\left( {{\alpha }_{2}}\left( B \right) \right)$. Furthermore, if $\mathbb{D}= {{\mathbb{R}}^{n}}$ and ${{\alpha }_{1}}\left( \centerdot  \right),{{\alpha }_{2}}\left( \centerdot  \right),{{\alpha }_{3}}\left( \centerdot  \right)$ are class $\mathcal{K_\infty}$ functions, the system is globally UUB.  
\end{theorem}

\begin{lemma}
\label{lm1}
\cite{c32} The following inequality holds for arbitrary ${{z}_{1}}$, ${{z}_{2}}$ $\in \mathbb{R}$, and any real-valued function $\mu \left( {{z}_{1}},{{z}_{2}} \right)\in {{\mathbb{R}}_{+}}\backslash \left\{ 0 \right\}$:  
\begin{equation}
\label{eq4}
        {{\left| {{z}_{1}} \right|}^{{{\rho }_{1}}}}{{\left| {{z}_{2}} \right|}^{{{\rho}_{2}}}}\le \frac{{{\rho}_{1}}{{\mu}}}{{{\rho }_{1}}+{{\rho}_{2}}}{{\left| {{z }_{1}} \right|}^{{{\rho}_{1}}+{{\rho}_{2}}}}+\frac{{{\rho}_{2}}{{\mu}}^{-\frac{{{\rho}_{1}}}{{{\rho}_{2}}}}}{{{\rho}_{1}}+{{\rho}_{2}}}{{\left| {{z }_{2}} \right|}^{{{\rho}_{1}}+{{\rho }_{2}}}},
\end{equation}
where $\rho_1$ and $\rho_2$ are positive real numbers.
\end{lemma}

\begin{lemma}
\label{lm2}
\cite{c33} The following inequality holds for any real number ${{z}_{i}},\,\,i={1,\ldots ,n}$ and $0<\rho_3 \le1$:  
\begin{equation}
\label{eq5}
\sum\limits_{i=1}^{n}{{{\left| {{z}_{i}} \right|}^{\rho_3}}}\ge {{\left( \sum\limits_{i=1}^{n}{\left| {{z}_{i}} \right|} \right)}^{\rho_3}}.
\end{equation}
\end{lemma}

% \begin{lemma}
% \label{lm3}
% (Cauchy--Schwarz Inequality) \cite{c54} For arbitrary $\mathbf{A},\,\,\mathbf{B}\in {{R}^{n}} $,
% the following inequality holds:  
% \begin{equation}
% \label{eq7}
% {{\left( {{\mathbf{A}}^{\top}}\mathbf{B} \right)}^{2}}\le \left( {{\mathbf{A}}^{\top}}\mathbf{A} \right)\left( {{\mathbf{B}}^{\top}}\mathbf{B} \right)
% \end{equation}
% \end{lemma}

\section{Main Results}
\label{sec3}
In this section, a robust time-varying SPVFC is designed to follow the desired trajectory, which is encoded by the concept of time-varying velocity field. Notably, the features of the proposed approach, ensuring task performance and safe interaction, are analyzed from theoretical viewpoints.

\subsection{Robust Time-varying SPVFC}
\label{3a}
In the present study, the fictitious flywheel concept is utilized to store and supplement kinetic energy during physical interactions. In this manner, an augmented state is incorporated with the dynamics model \eqref{eq1} as follows: 
\begin{equation}
\label{eq6}
{{\mathbf{M}}^{a}}\left( {{\mathbf{q}}^{a}} \right){{\mathbf{\ddot{q}}}^{a}}+{{\mathbf{C}}^{a}}\left( {{\mathbf{q}}^{a}},{{{\mathbf{\dot{q}}}}^{a}} \right){{\mathbf{\dot{q}}}^{a}}={{\pmb{\tau }}^{a}}+\pmb{\tau }_{ext}^{a},
\end{equation}
where ${{\mathbf{q}}^{a}}={{\left[ \begin{matrix}
   {{\mathbf{q}}^{\top}} & {{q}_{f}}  \\
\end{matrix} \right]}^{\top}}\in {{\mathbb{R}}^{n+1}}
$, $\mathbf{q}$, and ${{q}_{f}}$ are the states of the augmented system, mechanical system, and flywheel, respectively. ${{\pmb{\tau }}^{a}}={{\left[ \begin{matrix}
   {{\pmb{\tau }}^{\top}} & {{\tau }_{f}}  \\
\end{matrix} \right]}^{\top}}$  and $\pmb{\tau }_{ext}^{a}={{\left[ \begin{matrix}
   \pmb{\tau }_{ext}^{\top} & 0  \\
\end{matrix} \right]}^{\top}}$ are the augmented control input and external force vector, respectively. ${{\mathbf{M}}^{a}}\left( {{\mathbf{q}}^{a}} \right):=\left[ \begin{matrix}
   \mathbf{M}\left( \mathbf{q} \right) & {{\mathbf{0}}_{n\times 1}}  \\
   {{\mathbf{0}}_{1\times n}}  & {{m}_{f}}  \\
\end{matrix} \right]$ is the augmented inertia matrix, which is established by the inertia matrix of the mechanical system $\mathbf{M}\left( \mathbf{q} \right)$ and the mass of the flywheel ${{m}_{f}}$. ${{\mathbf{C}}^{a}}\left( {{\mathbf{q}}^{a}},{{{\mathbf{\dot{q}}}}^{a}} \right):=\left[ \begin{matrix}
   \mathbf{C}\left( \mathbf{q},\mathbf{\dot{q}} \right) & {{\mathbf{0}}_{n\times 1}}  \\
   {{\mathbf{0}}_{1\times n}} & 0  \\
\end{matrix} \right]$ with $\mathbf{C}\left( \mathbf{q},\mathbf{\dot{q}} \right)$ denotes the Coriolis and centrifugal matrix of the augmented system. $\left( {{{\mathbf{\dot{M}}}}^{a}}\left( {{\mathbf{q}}^{a}} \right)-2{{\mathbf{C}}^{a}}\left( {{\mathbf{q}}^{a}},{{{\mathbf{\dot{q}}}}^{a}} \right) \right)$ is also a skew-symmetric matrix.

To encode a desired time-variant trajectory $\mathbf{q}_{d}\left( t \right)$, we propose a time-varying velocity field instead of the time-invariant velocity fields used in \cite{c22,c23,c24,c25,c26}. The desired time-varying velocity field for the augmented system $ {{\mathbf{V}}^{a}}\left( {{\mathbf{q}}^{a}},t \right):={{\left[ \begin{matrix}
   {{\mathbf{V}}^{\top}}\left( \mathbf{q},t \right) & {{V}_{f}}\left( \mathbf{q},t \right)  \\
\end{matrix} \right]}^{\top}}$ is defined as:
\begin{equation}
\left\{ \begin{aligned}
  & \mathbf{V}\left( \mathbf{q},t \right):=\mathbf{\dot{q}}_{d}\left( t \right)-\pmb{\psi}\left( \mathbf{q}\left( t \right)-\mathbf{q}_{d}\left( t \right) \right) \\ 
 & {{V}_{f}}\left( \mathbf{q},t \right):=\sqrt{\frac{2}{{{m}_{f}}}\left( {{E}^{a}}-\frac{1}{2}{{\mathbf{V}}^{\top}}\left( \mathbf{q},t \right)\mathbf{M}\left( \mathbf{q} \right)\mathbf{V}\left( \mathbf{q},t \right) \right)} \\ 
\end{aligned} \right.
\label{eq7}
\end{equation}
where $\mathbf{V}\left( \mathbf{q},t \right)$ and ${{V}_{f}}\left( \mathbf{q},t \right)$ denote the desired time-varying velocity field of the mechanical system and flywheel, respectively. $\pmb{\psi}$ is the control parameter matrix. To derive the second equation in \eqref{eq7}, the principle of conservation of kinetic energy is employed, expressed as ${{k}^{a}}\left( {{\mathbf{q}}^{a}},{{\mathbf{V}}^{a}}\left( {{\mathbf{q}}^{a}},t \right) \right):=\frac{1}{2}{{\mathbf{V}}^{\top}}\left( \mathbf{q},t \right)\mathbf{M}\left( \mathbf{q} \right)\mathbf{V}\left( \mathbf{q},t \right)+\frac{1}{2}{{m}_{f}}V_{f}^{2}\left( \mathbf{q},t \right)={{E}^{a}}>0$, where ${{E}^{a}}$ is designed such that the formulation under the square root is positive. It should be noted that the time-varying velocity field in \eqref{eq7} is formulated to address timed trajectory tracking tasks during interaction with the physical environment, whereas existing studies \cite{c22,c23,c24,c25,c26} primarily consider time-invariant velocity fields for contour-following tasks.

%Obviously, the above feature of the second equation in \eqref{eq7} is violated without the augmented state of the fictitious flywheel.

Inspired by the preliminary version \cite{c27}, an improved robust time-varying SPVFC is developed for regulating the passivity domain, guaranteeing stability, as well as constraining energy level and power flow. By incorporating energy-compensation control terms, including integer and fractional-order elements, the control law is designed as:
\begin{equation}
\begin{aligned}
  & {{\pmb{\tau }}^{a}}\left( {{\mathbf{q}}^{a}},{{{\mathbf{\dot{q}}}}^{a}},t \right)={{\mathbf{R}}_{1}}\left( {{\mathbf{q}}^{a}},{{{\mathbf{\dot{q}}}}^{a}},t \right){{{\mathbf{\dot{q}}}}^{a}}+{{\mathbf{R}}_{2}}\left( {{\mathbf{q}}^{a}},{{{\mathbf{\dot{q}}}}^{a}},t \right){{{\mathbf{\dot{q}}}}^{a}} \\ 
 & \,\,\,\,\,\,\,\,\,\,\,\,\,\,\,\,\,\,\,\,\,\,\,\,\,\,\,\,\,\,\,\,\,\,\,\,\,\,\,\,-{{\mathbf{S}}_{1}}\left( {{\mathbf{q}}^{a}},{{{\mathbf{\dot{q}}}}^{a}} \right){{{\mathbf{\dot{q}}}}^{a}}-{{\mathbf{S}}_{2}}\left( {{\mathbf{q}}^{a}},{{{\mathbf{\dot{q}}}}^{a}} \right){{\left\lfloor {{{\mathbf{\dot{q}}}}^{a}} \right\rceil}^{\frac{{{\zeta}_{1}}}{{{\zeta}_{2}}}}}, \\ 
\end{aligned}
\label{eq8}
\end{equation}
in which ${{\zeta}_{1}},\,{{\zeta}_{2}}$ $\left( {{\zeta}_{1}}<{{\zeta}_{2}} \right)$ are  positive odd integers. ${{\mathbf{R}}_{1}}\left( {{\mathbf{q}}^{a}},{{{\mathbf{\dot{q}}}}^{a}},t \right):=\frac{1}{2{{E}^{a}}}\left( \mathbf{w}{{\mathbf{P}}^{\top}}-\mathbf{P}{{\mathbf{w}}^{\top}} \right)$ and ${{\mathbf{R}}_{2}}\left( {{\mathbf{q}}^{a}},{{{\mathbf{\dot{q}}}}^{a}},t \right):=\kappa \left( \mathbf{P}{{\mathbf{p}}^{\top}}-\mathbf{p}{{\mathbf{P}}^{\top}} \right)$ are skew-symmetric matrices with $\mathbf{w}:={{\mathbf{M}}^{a}}\left( {{\mathbf{q}}^{a}} \right){{{\nabla }_{t}{\mathbf{{V}}}}^{a}\left( \mathbf{q},t \right)}+{{\mathbf{C}}^{a}}\left( {{\mathbf{q}}^{a}},{{{\mathbf{\dot{q}}}}^{a}} \right){{\mathbf{V}}^{a}}\left( \mathbf{q},t \right)$, $\mathbf{P}:={{\mathbf{M}}^{a}}\left( {{\mathbf{q}}^{a}} \right){{\mathbf{V}}^{a}}\left( \mathbf{q},t \right)$, and $\mathbf{p}:={{\mathbf{M}}^{a}}\left( {{\mathbf{q}}^{a}} \right){{{\mathbf{\dot{q}}}}^{a}}$. In addition, ${{\mathbf{S}}_{1}}\left( {{\mathbf{q}}^{a}},{{{\mathbf{\dot{q}}}}^{a}} \right):=s\left( e_E \right){{\mathbf{K}}_{1}}$ and ${{\mathbf{S}}_{2}}\left( {{\mathbf{q}}^{a}},{{{\mathbf{\dot{q}}}}^{a}} \right):=s\left( e_E \right){{\mathbf{K}}_{2}}$ are matrices constructed to regulate energy level and power flow. $\kappa$, ${{\mathbf{K}}_{1}}$, and ${{\mathbf{K}}_{2}}$ represent control parameters. Furthermore, the smooth saturation function illustrated in Fig. \ref{fig0} is introduced as follows:
\begin{equation}
s\left( e_E \right):=\left\{ \begin{aligned}
  & \,\,\,\,\,\,\,\,\,\,\,\,\,\,\,\,\,-{{\eta }_{\min }}\,\,\,\,\,\,\,\,\,\,\,\,\,\,\,\,\,\,\,\,\,\,\,\,\,\,\,\,\,\,\,\,\,\,\,\,\,e_E<-{{\delta }_{1}}-{{\delta }_{2}} \\ 
 & \phi \left( -{{\eta }_{\min }},\frac{{{e}_{E}}+{{\delta }_{2}}}{{{\delta }_{1}}} \right)\,\,\,\,\,{{-\delta }_{1}}-{{\delta }_{2}}\le e_E<{{-\delta }_{2}} \\ 
 & \,\,\,\,\,\,\,\,\,\,\,\,\,\,\,\,\,\,\,\,0\,\,\,\,\,\,\,\,\,\,\,\,\,\,\,\,\,\,\,\,\,\,\,\,\,\,\,\,\,\,\,\,\,\,\,\,\,\,\,\,\,\,\,\,\,{{-\delta }_{2}}\le e_E\le {{\delta }_{3}} \\ 
 & \phi \left( {{\eta }_{\max }},\frac{{{e}_{E}}-{{\delta }_{3}}}{{{\delta }_{4}}} \right)\,\,\,\,\,\,\,\,\,\,\,\,\,\,\,{{\delta }_{3}}<e_E\le{{\delta }_{3}}+{{\delta }_{4}} \\ 
 &\,\,\,\,\,\,\,\,\,\,\,\,\,\,\,\,\,{{\eta }_{\max }}\,\,\,\,\,\,\,\,\,\,\,\,\,\,\,\,\,\,\,\,\,\,\,\,\,\,\,\,\,\,\,\,\,\,\,\,\,\,\,\,\,\,e_E> {{\delta }_{3}}+{{\delta }_{4}} \\ 
\end{aligned} \right.
\label{eq9}
\end{equation}
in which $\phi \left( *,\circ  \right):=\frac{*}{2}\left( 1-\cos \left( \pi \circ  \right) \right)$ and $e_E:={{k}^{a}}\left( {{\mathbf{q}}^{a}},{{{\mathbf{\dot{q}}}}^{a}} \right)-k_{d}^{a}$. ${{k}^{a}}\left( {{\mathbf{q}}^{a}},{{{\mathbf{\dot{q}}}}^{a}} \right):=\frac{1}{2}{{\mathbf{\dot{q}}}^{\top}}\mathbf{M}\left( \mathbf{q} \right)\mathbf{\dot{q}}+\frac{1}{2}{{m}_{f}}\dot{q}_{f}^{2}$ and $k_{d}^{a}$ are the actual and desired kinetic energy level of the augmented system. ${{\eta }_{\min }}, {{\eta }_{\max }}>0$, ${{\delta }_{1}},\,{{\delta }_{2}}>0$ with ${{\delta }_{1}}+{{\delta }_{2}}<k_{d}^{a}$, and ${{\delta }_{3}},\,{{\delta }_{4}}>0$ denote the designed constant values. 

%Thus, by incorporating two compensated energy terms with the smooth saturation function, the robust time-varying SPVFC aims to drive and maintain the energy error $e_E$ within the designed operating interval $\left[ {{\delta }_{2}},{{\delta }_{3}} \right]$.

\begin{figure}[t]
    \centering
\includegraphics[width=0.65\linewidth]{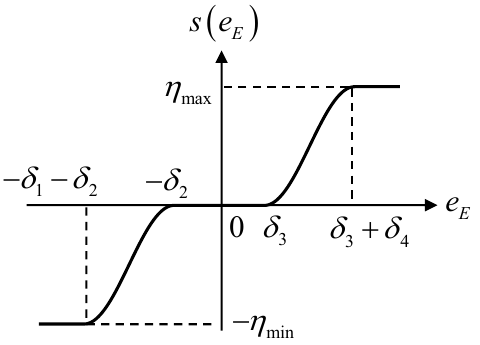}
    \caption{Illustration of  the smooth saturation function \eqref{eq9}.}
    \label{fig0}
\end{figure}

\begin{remark}
\label{remark0}
Although the use of only the first two terms in \eqref{eq8}, as in prior works \cite{c22,c23,c24}, guarantees energetic passivity of the closed-loop system with respect to the force–velocity input–output pair, this approach tends to be conservative and degraded task performance under energy-depleted conditions. Therefore, the present study proposes the SPVFC in \eqref{eq8} by integrating the two terms adopted from existing methods \cite{c22,c23,c24} with two additional control terms ${{\mathbf{S}}_{1}}\left( {{\mathbf{q}}^{a}},{{{\mathbf{\dot{q}}}}^{a}} \right){{{\mathbf{\dot{q}}}}^{a}}, \,{{\mathbf{S}}_{2}}\left( {{\mathbf{q}}^{a}},{{{\mathbf{\dot{q}}}}^{a}} \right){{\left\lfloor {{{\mathbf{\dot{q}}}}^{a}} \right\rceil}^{\frac{{{\zeta}_{1}}}{{{\zeta}_{2}}}}}$ to balance task performance and physical interaction safety.  Specifically, as illustrated in Fig. \ref{fig0}, when $e_E\ge\delta_3$, the system actively dissipates energy to mitigate excessive accumulation and prevent potentially hazardous interactions, while preserving passivity. Conversely, when $e_E\le-\delta_2$, energy is injected into the closed-loop system to avoid energy depletion and ensure the completion of the desired tasks. By incorporating two compensated energy terms, the robust time-varying SPVFC aims to drive and maintain the energy error $e_E$ within the designed operating interval $\left[ {{-\delta }_{2}},{{\delta }_{3}} \right]$. Furthermore, the inclusion of both ${{\mathbf{S}}_{1}}\left( {{\mathbf{q}}^{a}},{{{\mathbf{\dot{q}}}}^{a}} \right){{{\mathbf{\dot{q}}}}^{a}}$ and ${{\mathbf{S}}_{2}}\left( {{\mathbf{q}}^{a}},{{{\mathbf{\dot{q}}}}^{a}} \right){{\left\lfloor {{{\mathbf{\dot{q}}}}^{a}} \right\rceil}^{\frac{{{\zeta}_{1}}}{{{\zeta}_{2}}}}}$ in the controller \eqref{eq8} improves energy regulation performance, where the linear term enhances the compensation rate and the nonlinear fractional-power term ensures energy compensation in finite time. 
\end{remark}

By applying the proposed robust time-varying SPVFC \eqref{eq8} for the augmented mechanical system \eqref{eq6}, the following theorems and lemmas hold.

\subsection{Passivity and its conditions}
\label{3b}
\begin{theorem}
\label{thm2}
Given the augmented system \eqref{eq6}, the robust time-varying SPVFC law \eqref{eq8}, positive odd integers ${{\zeta}_{1}},\,{{\zeta}_{2}}$ $\left( {{\zeta}_{1}}<{{\zeta}_{2}} \right)$, and positive-definite parameter matrices ${{\mathbf{K}}_{1}}$, ${{\mathbf{K}}_{2}}$, the closed-loop mechanical system is passive with respect to the input–output pair $\pmb{\tau }_{ext}^{a}$ and ${{{\mathbf{\dot{q}}}}^{a}}$ as long as the energy constraint ${{k}^{a}}\left( {{\mathbf{q}}^{a}},{{{\mathbf{\dot{q}}}}^{a}} \right)\ge k_{d}^{a}-{{\delta }_{2}}$ is satisfied for all times.

\textit{Proof:} Differentiating the energy storage function ${{k}^{a}}\left( {{\mathbf{q}}^{a}},{{{\mathbf{\dot{q}}}}^{a}} \right)=\frac{1}{2}{{{\mathbf{\dot{q}}}}^{a\top}}{{\mathbf{M}}^{a}}\left( {{\mathbf{q}}^{a}} \right){{{\mathbf{\dot{q}}}}^{a}}$, combining with the augmented system \eqref{eq6}, and adopting the control law \eqref{eq8}, we obtain:
\begin{equation}
\label{eq10}
\begin{aligned}
  & \frac{d}{dt}{{k}^{a}}\left( {{\mathbf{q}}^{a}}\left( t \right),{{{\mathbf{\dot{q}}}}^{a}}\left( t \right) \right)\\ & =
  \frac{1}{2}{{{\mathbf{\dot{q}}}}^{a\top}}{{{\mathbf{\dot{M}}}}^{a}}\left( {{\mathbf{q}}^{a}} \right){{{\mathbf{\dot{q}}}}^{a}} +{{{\mathbf{\dot{q}}}}^{a\top}}\left( {{\pmb{\tau }}^{a}}+\pmb{\tau }_{ext}^{a}-{{\mathbf{C}}^{a}}\left( {{\mathbf{q}}^{a}},{{{\mathbf{\dot{q}}}}^{a}} \right){{{\mathbf{\dot{q}}}}^{a}} \right) \\ 
 & ={{{\mathbf{\dot{q}}}}^{a\top}}\pmb{\tau }_{ext}^{a}+\frac{1}{2}{{{\mathbf{\dot{q}}}}^{a\top}}\left({{{\mathbf{\dot{M}}}}^{a}}\left( {{\mathbf{q}}^{a}} \right)-2{{\mathbf{C}}^{a}}\left( {{\mathbf{q}}^{a}},{{{\mathbf{\dot{q}}}}^{a}} \right) \right){{{\mathbf{\dot{q}}}}^{a}} \\ 
 &\,\,\,\,\,\,+{{{\mathbf{\dot{q}}}}^{a\top}}{{\mathbf{R}}_{1}}\left( {{\mathbf{q}}^{a}},{{{\mathbf{\dot{q}}}}^{a}},t \right){{{\mathbf{\dot{q}}}}^{a}}+{{{\mathbf{\dot{q}}}}^{a\top}}{{\mathbf{R}}_{2}}\left( {{\mathbf{q}}^{a}},{{{\mathbf{\dot{q}}}}^{a}},t \right){{{\mathbf{\dot{q}}}}^{a}} \\ 
 &\,\,\,\,\,\,-{{{\mathbf{\dot{q}}}}^{a\top}}{{\mathbf{S}}_{1}}\left( {{\mathbf{q}}^{a}},{{{\mathbf{\dot{q}}}}^{a}} \right){{{\mathbf{\dot{q}}}}^{a}}-{{{\mathbf{\dot{q}}}}^{a\top}}{{\mathbf{S}}_{2}}\left( {{\mathbf{q}}^{a}},{{{\mathbf{\dot{q}}}}^{a}} \right){{\left\lfloor {{{\mathbf{\dot{q}}}}^{a}} \right\rceil}^{\frac{{{\zeta}_{1}}}{{{\zeta}_{2}}}}}. 
\end{aligned}
\end{equation}
Note that all ${{{\mathbf{\dot{M}}}}^{a}}\left( {{\mathbf{q}}^{a}} \right)-2{{\mathbf{C}}^{a}}\left( {{\mathbf{q}}^{a}},{{{\mathbf{\dot{q}}}}^{a}} \right)$, ${{\mathbf{R}}_{1}}\left( {{\mathbf{q}}^{a}},{{{\mathbf{\dot{q}}}}^{a}},t \right)$, and ${{\mathbf{R}}_{2}}\left( {{\mathbf{q}}^{a}},{{{\mathbf{\dot{q}}}}^{a}},t \right)$ are skew symmetric; therefore, \eqref{eq10} becomes:
\begin{equation}
\label{eq11}
\begin{aligned}
  & \frac{d}{dt}{{k}^{a}}\left( {{\mathbf{q}}^{a}}\left( t \right),{{{\mathbf{\dot{q}}}}^{a}}\left( t \right) \right) ={{{\mathbf{\dot{q}}}}^{a\top}}\pmb{\tau }_{ext}^{a}-{{{\mathbf{\dot{q}}}}^{a\top}}{{\mathbf{S}}_{1}}\left( {{\mathbf{q}}^{a}},{{{\mathbf{\dot{q}}}}^{a}} \right){{{\mathbf{\dot{q}}}}^{a}}\\ & \,\,\,\,\,\,\,\,\,\,\,\,\,\,\,\,\,\,\,\,\,\,\,\,\,\,\,\,\,\,\,\,\,\,\,\,\,\,\,\,\,\,\,\,\,\,\,\,\,\,\,\,\,\,\,\,-{{{\mathbf{\dot{q}}}}^{a\top}}{{\mathbf{S}}_{2}}\left( {{\mathbf{q}}^{a}},{{{\mathbf{\dot{q}}}}^{a}} \right){{\left\lfloor {{{\mathbf{\dot{q}}}}^{a}} \right\rceil}^{\frac{{{\zeta}_{1}}}{{{\zeta}_{2}}}}}. 
\end{aligned}
\end{equation}
By calculating the integration of both sides of \eqref{eq11}, we have:
\begin{equation}
\begin{aligned}
  &\int\limits_{0}^{t}{{{{\mathbf{\dot{q}}}}^{a\top}}\left( \varepsilon\right)\pmb{\tau }_{ext}^{a}\left(\varepsilon  \right)d\varepsilon }-{{D}_{1}}-{{D}_{2}}={{k}^{a}}\left( {{\mathbf{q}}^{a}}\left( t \right),{{{\mathbf{\dot{q}}}}^{a}}\left( t \right) \right) \\ 
 &\,\,\,\,\,\,\,\,\,\,\,\,\,\,\,\,\,\,\,\,\,\,\,\,\,\,\,\,\,\,\,\,\,\,\,\,\,\,\,\,\,\,\,\,\,\,\,\,\,\,\,\,\,\,\,\,\,\,\,\,\,\,\,\,\,\,\,\,\,\,\,\,\,\,\,\,\,\,\,\,\,\,\,\,\,\,\,\,\, -{{k}^{a}}\left( {{\mathbf{q}}^{a}}\left( 0 \right),{{{\mathbf{\dot{q}}}}^{a}}\left( 0 \right) \right).
\end{aligned}
\label{eq12}
\end{equation}
If ${{\zeta}_{1}},\,{{\zeta}_{2}}$ $\left( {{\zeta}_{1}}<{{\zeta}_{2}} \right)$ are  positive odd integers and ${{\mathbf{K}}_{1}},\,{{\mathbf{K}}_{2}}$ are positive define matrices,
${{D}_{1}}=\int\limits_{0}^{t}{s\left( {{k}^{a}}\left( {{\mathbf{q}}^{a}}\left( \varepsilon\right),{{{\mathbf{\dot{q}}}}^{a}}\left(\varepsilon\right) \right)-k_{d}^{a} \right){{{\mathbf{\dot{q}}}}^{a\top}}\left(\varepsilon \right){{\mathbf{K}}_{1}}{{{\mathbf{\dot{q}}}}^{a}}\left(\varepsilon \right)d\varepsilon}$ and
${{D}_{2}}=\int\limits_{0}^{t}{s\left( {{k}^{a}}\left( {{\mathbf{q}}^{a}}\left( \varepsilon\right),{{{\mathbf{\dot{q}}}}^{a}}\left(\varepsilon  \right) \right)-k_{d}^{a} \right){{{\mathbf{\dot{q}}}}^{a\top}}\left(\varepsilon  \right){{\mathbf{K}}_{2}}{{\left\lfloor {{{\mathbf{\dot{q}}}}^{a}}\left(\varepsilon  \right) \right\rceil}^{\frac{{{\zeta}_{1}}}{{{\zeta}_{2}}}}}d\varepsilon }$ are  non-negative functions when ${{k}^{a}}\left( {{\mathbf{q}}^{a}},{{{\mathbf{\dot{q}}}}^{a}} \right)\ge k_{d}^{a}-{{\delta }_{2}}$. Indeed, ${{k}^{a}}\left( {{\mathbf{q}}^{a}}\left( t \right),{{{\mathbf{\dot{q}}}}^{a}}\left( t \right) \right)$ is consistently positive; therefore, \eqref{eq12} is presented by:
\begin{equation}
-{{k}^{a}}\left( {{\mathbf{q}}^{a}}\left( 0 \right),{{{\mathbf{\dot{q}}}}^{a}}\left( 0 \right) \right)\le \int\limits_{0}^{t}{{{{\mathbf{\dot{q}}}}^{a\top}}\left(\varepsilon  \right)\pmb{\tau }_{ext}^{a}\left(\varepsilon  \right)d\varepsilon }-{{D}_{1}}-{{D}_{2}}
\label{eq13}
\end{equation}
Following Definition \ref{def1}, the closed-loop mechanical system is passive with the pair of $\left(\pmb{\tau }_{ext}^{a},{{{\mathbf{\dot{q}}}}^{a}}\right)$ and the storage function ${{k}^{a}}\left( {{\mathbf{q}}^{a}},{{{\mathbf{\dot{q}}}}^{a}} \right)$ if and only if ${{k}^{a}}\left( {{\mathbf{q}}^{a}},{{{\mathbf{\dot{q}}}}^{a}} \right)\ge k_{d}^{a}-{{\delta }_{2}}$. This completes the proof. \hfill $\blacksquare$
\end{theorem}

\begin{remark}
\label{remark1}
According to the definition of the smooth continuous saturation function in \eqref{eq9} and Fig. \ref{fig0}, ${{D}_{1}},\,{{D}_{2}}>0$ if ${{k}^{a}}\left( {{\mathbf{q}}^{a}},{{{\mathbf{\dot{q}}}}^{a}} \right)> k_{d}^{a}+{{\delta }_{3}}$; ${{D}_{1}},\,{{D}_{2}}<0$ if ${{k}^{a}}\left( {{\mathbf{q}}^{a}},{{{\mathbf{\dot{q}}}}^{a}} \right)< k_{d}^{a}-{{\delta }_{2}}$; and otherwise ${{D}_{1}},\,{{D}_{2}}=0$. In this manner, the closed-loop mechanical system exhibits passive behavior when ${{k}^{a}}\left( {{\mathbf{q}}^{a}},{{{\mathbf{\dot{q}}}}^{a}} \right)\ge k_{d}^{a}-{{\delta }_{2}}$, as established in Theorem \ref{thm2}, whereas non-passive behavior is permitted when ${{k}^{a}}\left( {{\mathbf{q}}^{a}},{{{\mathbf{\dot{q}}}}^{a}} \right)< k_{d}^{a}-{{\delta }_{2}}$ to inject energy and maintain task performance. Thus, the continuously smooth saturation function \eqref{eq9} in ${{\mathbf{S}}_{1}}\left( {{\mathbf{q}}^{a}},{{{\mathbf{\dot{q}}}}^{a}} \right){{{\mathbf{\dot{q}}}}^{a}}=s\left( e_E \right){{\mathbf{K}}_{1}}{{{\mathbf{\dot{q}}}}^{a}}$ and ${{\mathbf{S}}_{2}}\left( {{\mathbf{q}}^{a}},{{{\mathbf{\dot{q}}}}^{a}} \right){{\left\lfloor {{{\mathbf{\dot{q}}}}^{a}} \right\rceil}^{\frac{{{\zeta}_{1}}}{{{\zeta}_{2}}}}}=s\left( e_E \right){{\mathbf{K}}_{2}}{{\left\lfloor {{{\mathbf{\dot{q}}}}^{a}} \right\rceil}^{\frac{{{\zeta}_{1}}}{{{\zeta}_{2}}}}}$ of the controller \eqref{eq8} alleviates conservatism, enabling a smooth transition between passive and non-passive behaviors. Notice that the passive and non-passive domains can be adjusted and configured by selecting the appropriate parameters $k_{d}^{a}$ and ${{\delta }_{2}}$. 

%As a result, following Theorem \ref{thm2}, the amount of energy is dissipated by the positive definite terms ${{D}_{1}}$ and ${{D}_{2}}$ when the energy level exceeds the given threshold ${{k}^{a}}\left( {{\mathbf{q}}^{a}},{{{\mathbf{\dot{q}}}}^{a}} \right)> k_{d}^{a}+{{\delta }_{3}}$. Meanwhile, the closed-loop mechanical system is supplemented by the energy from ${{D}_{1}}$ and ${{D}_{2}}$ in the case of ${{k}^{a}}\left( {{\mathbf{q}}^{a}},{{{\mathbf{\dot{q}}}}^{a}} \right)< k_{d}^{a}+{{\delta }_{2}}$ to perform the desired task. Thus, the SPVFC control law in \eqref{eq8} relaxes the conservative nature by using smooth continuous transition between passive and non-passive behaviors in a controlled manner. Notice that the passive and non-passive domains can be adjusted and configured by selecting the appropriate parameters $k_{d}^{a}$, ${{\delta }_{2}}$, and ${{\delta }_{3}}$.
\end{remark}

From Theorem \ref{thm2} and Remarks \ref{remark0}, \ref{remark1}, the system's energy tends toward the region $\left[ k_{d}^{a}-{{\delta }_{2}},k_{d}^{a}+{{\delta }_{3}} \right]$ due to two terms $D_1$ and $D_2$. Theorem \ref{thm3} in the next subsection will investigate this trend under external forces from physical environments to present tightness conditions and energy-compensation rate.

\subsection{Energy Boundedness Analysis}
\begin{theorem}
\label{thm3}
Considering the augmented mechanical system with external disturbance \eqref{eq6} under the control action \eqref{eq8}, the kinetic energy level of the closed-loop system converges to a region $B_1$ defined by:
\begin{equation}
\label{eq14}
{{B}_{1}}:=\left\{ {{k}^{a}}\left( {{\mathbf{q}}^{a}},{{{\mathbf{\dot{q}}}}^{a}} \right)\in {{\mathbb{R}}_{+}}\left| \begin{aligned}
  & {{k}^{a}}\left( {{\mathbf{q}}^{a}},{{{\mathbf{\dot{q}}}}^{a}} \right)\ge k_{d}^{a}-{{\delta }_{2}} \\ 
 & {{k}^{a}}\left( {{\mathbf{q}}^{a}},{{{\mathbf{\dot{q}}}}^{a}} \right)\le k_{d}^{a}+{{\delta }_{3}} \\ 
\end{aligned} \right. \right\},
\end{equation}
after an explicit finite settling time ${{T}_{E}}\le \max \left\{ {{T}_{1}},{{T}_{2}} \right\}$. Here,
\begin{equation}
\label{eq25}
{{T}_{1}}\le\frac{2{{\zeta }_{2}}}{{{\gamma }_{1}}\left( {{\zeta }_{2}}-{{\zeta }_{1}} \right)}\ln \left( \frac{{{\gamma }_{1}}{{\left( {{k}^{a}}\left( {{\mathbf{q}}^{a}}\left( 0 \right),{{{\mathbf{\dot{q}}}}^{a}}\left( 0 \right) \right) \right)}^{\frac{{{\zeta }_{2}}-{{\zeta }_{1}}}{2{{\zeta }_{2}}}}}+{{\gamma }_{4}}}{{{\gamma }_{1}}{{\left( k_{d}^{a}+\delta_3 \right)}^{\frac{{{\zeta }_{2}}-{{\zeta }_{1}}}{2{{\zeta }_{2}}}}}+{{\gamma }_{4}}} \right),
\end{equation}
and 
\begin{equation}
\label{eq27}
{{T}_{2}}\le \frac{2{{\zeta }_{2}}}{{{\gamma }_{1}}\left( {{\zeta }_{2}}-{{\zeta }_{1}} \right)}\ln \left( \frac{{{\gamma }_{1}}{{\left( k_{d}^{a}-\delta_2 \right)}^{\frac{{{\zeta }_{2}}-{{\zeta }_{1}}}{2{{\zeta }_{2}}}}}+{{\gamma }_{4}}}{{{\gamma }_{1}}{{\left( {{k}^{a}}\left( {{\mathbf{q}}^{a}}\left( 0 \right),{{{\mathbf{\dot{q}}}}^{a}}\left( 0 \right) \right) \right)}^{\frac{{{\zeta }_{2}}-{{\zeta }_{1}}}{2{{\zeta }_{2}}}}}+{{\gamma }_{4}}} \right),
\end{equation}
denote the convergence times for ${{k}^{a}}\left( {{\mathbf{q}}^{a}},{{{\mathbf{\dot{q}}}}^{a}} \right)>k_{d}^{a}+{{\delta }_{3}}$ and ${{k}^{a}}\left( {{\mathbf{q}}^{a}},{{{\mathbf{\dot{q}}}}^{a}} \right)<k_{d}^{a}-{{\delta }_{2}}$, respectively, where the auxiliary control parameters $\gamma_1$ and $\gamma_4$ are proportional to $\mathbf{K}_1$ and $\mathbf{K}_2$ in \eqref{eq8}. 

\textit{Proof:} Utilizing the continuously differentiable positive function ${{V}_{1}}\left( e_{E}^{\delta } \right):=\frac{1}{2}\left(e_{E}^{\delta }\right)^2=\frac{1}{2}{{\left( {{k}^{a}}\left( {{\mathbf{q}}^{a}},{{{\mathbf{\dot{q}}}}^{a}} \right)-\left(k_{d}^{a}+\delta\right) \right)}^{2}}$ with $\delta \in \left\{ -{{\delta }_{2}},{{\delta }_{3}} \right\}$, its derivative is computed and combined with the augmented mechanical system \eqref{eq6} as follows: 
\begin{equation}
\label{eq15}
\begin{aligned}
  & {{{\dot{V}}}_{1}}\left( e_{E}^{\delta } \right)=e_{E}^{\delta }\left( {{{\mathbf{\dot{q}}}}^{a\top}}{{\mathbf{M}}^{a}}\left( {{\mathbf{q}}^{a}} \right){{{\mathbf{\ddot{q}}}}^{a}}+\frac{1}{2}{{{\mathbf{\dot{q}}}}^{a\top}}{{{\mathbf{\dot{M}}}}^{a}}\left( {{\mathbf{q}}^{a}} \right){{{\mathbf{\dot{q}}}}^{a}} \right) \\ 
 & \,\,\,\,\,\,\,\,\,\,\,\,\,\,\,\,\,\,\,\,=e_{E}^{\delta }\left( \begin{aligned}
  & \frac{1}{2}{{{\mathbf{\dot{q}}}}^{a\top}}{{{\mathbf{\dot{M}}}}^{a}}\left( {{\mathbf{q}}^{a}} \right){{{\mathbf{\dot{q}}}}^{a}} \\ 
 & +{{{\mathbf{\dot{q}}}}^{a\top}}\left( {{\pmb{\tau }}^{a}}+\pmb{\tau }_{ext}^{a}-{{\mathbf{C}}^{a}}\left( {{\mathbf{q}}^{a}},{{{\mathbf{\dot{q}}}}^{a}} \right){{{\mathbf{\dot{q}}}}^{a}} \right) \\ 
\end{aligned} \right). \\ 
\end{aligned}
\end{equation}
Applying the robust time-varying SPVFC \eqref{eq8}, yields:
\begin{equation}
\label{eq16}
{{\dot{V}}_{1}}=e_{E}^{\delta }\left( \begin{aligned}
  & \frac{1}{2}{{{\mathbf{\dot{q}}}}^{a\top}}\left( {{{\mathbf{\dot{M}}}}^{a}}\left( {{\mathbf{q}}^{a}} \right)-2{{\mathbf{C}}^{a}}\left( {{\mathbf{q}}^{a}},{{{\mathbf{\dot{q}}}}^{a}} \right) \right){{{\mathbf{\dot{q}}}}^{a}}+{{{\mathbf{\dot{q}}}}^{a\top}}\pmb{\tau }_{ext}^{a} \\ 
 & +{{{\mathbf{\dot{q}}}}^{a\top}}{{\mathbf{R}}_{1}}\left( {{\mathbf{q}}^{a}},\,{{{\mathbf{\dot{q}}}}^{a}},t \right){{{\mathbf{\dot{q}}}}^{a}}+{{{\mathbf{\dot{q}}}}^{a\top}}{{\mathbf{R}}_{2}}\left( {{\mathbf{q}}^{a}},\,{{{\mathbf{\dot{q}}}}^{a}},t \right){{{\mathbf{\dot{q}}}}^{a}} \\ 
 & -{{{\mathbf{\dot{q}}}}^{a\top}}{{\mathbf{S}}_{1}}\left( {{\mathbf{q}}^{a}},\,{{{\mathbf{\dot{q}}}}^{a}} \right){{{\mathbf{\dot{q}}}}^{a}}-{{{\mathbf{\dot{q}}}}^{a\top}}{{\mathbf{S}}_{2}}\left( {{\mathbf{q}}^{a}},\,{{{\mathbf{\dot{q}}}}^{a}} \right){{\left\lfloor {{{\mathbf{\dot{q}}}}^{a}} \right\rceil}^{\frac{{{\zeta}_{1}}}{{{\zeta}_{2}}}}} \\ 
\end{aligned} \right)
\end{equation}
Due to the property of skew-symmetric matrices, $e_{E}^{\delta }s\left( e_E \right)\ge 0$, and $e_E\notin \left[ -{{\delta }_{2}},{{\delta }_{3}} \right]$, \eqref{eq16} can be rewritten as:
\begin{equation}
\label{eq17}
\begin{aligned}
  & {{{\dot{V}}}_{1}}=e_{E}^{\delta }{{{\mathbf{\dot{q}}}}^{a\top}}\pmb{\tau }_{ext}^{a}-e_{E}^{\delta }s\left( e_E \right){{{\mathbf{\dot{q}}}}^{a\top}}{{\mathbf{K}}_{1}}{{{\mathbf{\dot{q}}}}^{a}} \\ 
 & \,\,\,\,\,\,\,\,\,\,\,\,-e_{E}^{\delta }s\left( e_E \right){{{\mathbf{\dot{q}}}}^{a\top}}{{\mathbf{K}}_{2}}{{\left\lfloor {{{\mathbf{\dot{q}}}}^{a}} \right\rceil}^{\frac{{{\zeta}_{1}}}{{{\zeta}_{2}}}}} \\ 
 & \,\,\,\,\le \left| e_{E}^{\delta } \right|\left\| {{{\mathbf{\dot{q}}}}^{a\top}} \right\|{{\left\| \pmb{\tau }_{ext}^{a} \right\|}}-{s_{\min}}\left| e_{E}^{\delta } \right|{{{\mathbf{\dot{q}}}}^{a\top}}{{\mathbf{K}}_{1}}{{{\mathbf{\dot{q}}}}^{a}} \\ 
 & \,\,\,\,\,\,\,\,\,\,\,\,-{s_{\min}}\left|  e_{E}^{\delta } \right|{{{\mathbf{\dot{q}}}}^{a\top}}{{\mathbf{K}}_{2}}{{\left\lfloor {{{\mathbf{\dot{q}}}}^{a}} \right\rceil}^{\frac{{{\zeta}_{1}}}{{{\zeta}_{2}}}}}, \\ 
\end{aligned}
\end{equation}
in which ${s_{\min}}=\min \left( \left| s\left( {{e}_{E}} \right) \right| \right)$ in the case of $e_E\notin \left[- {{\delta }_{2}},{{\delta }_{3}} \right]$. Using the inequality in Lemma \ref{lm1} with $\mathbf{K_1}$, $\mathbf{K_2}$, and ${{\mathbf{M}}^{a}}\left( {{\mathbf{q}}^{a}} \right)$ established by trigonometric functions or constants, we obtain:
\begin{equation}
\label{eq18}
\left\{ \begin{aligned}
  & \left\| {{{\mathbf{\dot{q}}}}^{a\top}} \right\|{{\left\| \pmb{\tau }_{ext}^{a} \right\|}\le \frac{{{\mu_1\left\| {{{\mathbf{\dot{q}}}}^{a}} \right\|}^{2}}}{2}+\frac{\left\| \pmb{\tau }_{ext}^{a} \right\|^{2}}{2\mu_1}} \\ 
 & {{k}^{a}}\left( {{\mathbf{q}}^{a}},{{{\mathbf{\dot{q}}}}^{a}} \right)\ge \frac{1}{2}{{l}_{\min }}\left( {{\mathbf{M}}^{a}}\left( {{\mathbf{q}}^{a}} \right) \right){{\left\| {{{\mathbf{\dot{q}}}}^{a}} \right\|}^{2}} \\ 
 & {{k}^{a}}\left( {{\mathbf{q}}^{a}},{{{\mathbf{\dot{q}}}}^{a}} \right)\le \frac{1}{2}{{l}_{\max }}\left( {{\mathbf{M}}^{a}}\left( {{\mathbf{q}}^{a}} \right) \right){{\left\| {{{\mathbf{\dot{q}}}}^{a}} \right\|}^{2}} \\ 
 & {{l}_{\min }}\left( {{\mathbf{K}}_{1}} \right){{\left\| {{{\mathbf{\dot{q}}}}^{a}} \right\|}^{2}}\le {{{\mathbf{\dot{q}}}}^{a\top}}{{\mathbf{K}}_{1}}{{{\mathbf{\dot{q}}}}^{a}}\le {{l}_{\max }}\left( {{\mathbf{K}}_{1}} \right){{\left\| {{{\mathbf{\dot{q}}}}^{a}} \right\|}^{2}} \\ 
 & {{{\mathbf{\dot{q}}}}^{a\top}}{{\mathbf{K}}_{2}}{{\left\lfloor {{{\mathbf{\dot{q}}}}^{a}} \right\rceil}^{\frac{{{\zeta}_{1}}}{{{\zeta}_{2}}}}}\ge {{l}_{\min }}\left( {{\mathbf{K}}_{2}} \right){{\left( {{\left\| {{{\mathbf{\dot{q}}}}^{a}} \right\|}_{\frac{{{\zeta }_{1}}+{{\zeta }_{2}}}{{{\zeta }_{2}}}}} \right)}^{\frac{{{\zeta }_{1}}+{{\zeta }_{2}}}{{{\zeta }_{2}}}}} \\ 
 & {{{\mathbf{\dot{q}}}}^{a\top}}{{\mathbf{K}}_{2}}{{\left\lfloor {{{\mathbf{\dot{q}}}}^{a}} \right\rceil}^{\frac{{{\zeta}_{1}}}{{{\zeta}_{2}}}}}\le {{l}_{\max }}\left( {{\mathbf{K}}_{2}} \right){{\left( {{\left\| {{{\mathbf{\dot{q}}}}^{a}} \right\|}_{\frac{{{\zeta }_{1}}+{{\zeta }_{2}}}{{{\zeta }_{2}}}}} \right)}^{\frac{{{\zeta }_{1}}+{{\zeta }_{2}}}{{{\zeta }_{2}}}}} \\ 
\end{aligned} \right.
\end{equation}
where $\mu_1$ is an arbitrary positive constant from the inequality in Lemma \ref{lm1}. ${{l}_{\min }}\left( \bullet  \right)$ and ${{l}_{\max }}\left( \bullet  \right)$ are the smallest and largest eigenvalues of the matrix $\left( \bullet  \right)$. Substituting \eqref{eq18} into \eqref{eq17} and adopting Lemma \ref{lm2}, yields:
\begin{equation}
\label{eq19}
\begin{aligned}
  & {{{\dot{V}}}_{1}}\,\le -\left( s_{\min}{{l}_{\min }}\left( {{\mathbf{K}}_{1}} \right)-\frac{\mu_1}{2} \right)\sqrt{2{{V}_{1}}}\left\| {{{\mathbf{\dot{q}}}}^{a}} \right\|^{2} \\ 
 & \,\,\,\,\,\,\,\,\,\,\,\,\,-s_{\min}{l_{\min}}\left( {{\mathbf{K}}_{2}} \right)\sqrt{2{{V}_{1}}}{{\left( \left\| {{{\mathbf{\dot{q}}}}^{a}} \right\|^{2} \right)}^{\frac{{{\zeta }_{1}}+{{\zeta }_{2}}}{2{{\zeta }_{2}}}}} +\sqrt{2{{V}_{1}}}\frac{\left\| \pmb{\tau }_{ext}^{a} \right\|^{2}}{2\mu_1} \\ 
 & \,\,\,\,\,\,\le -{{\gamma }_{1}}\sqrt{2{{V}_{1}}}{{k}^{a}}\left( {{\mathbf{q}}^{a}},{{{\mathbf{\dot{q}}}}^{a}} \right)-{{\gamma }_{2}}{{\gamma }_{3}}\sqrt{2{{V}_{1}}}{{k}^{a}}\left( {{\mathbf{q}}^{a}},{{{\mathbf{\dot{q}}}}^{a}} \right) \\ 
 &\,\,\,\,\,\,\,\,\,\,\,\,\,-{{\gamma }_{4}}\sqrt{2{{V}_{1}}}{{\left( {{k}^{a}}\left( {{\mathbf{q}}^{a}},{{{\mathbf{\dot{q}}}}^{a}} \right) \right)}^{\frac{{{\zeta }_{1}}+{{\zeta }_{2}}}{2{{\zeta }_{2}}}}}+\sqrt{2{{V}_{1}}}\frac{\left\| \pmb{\tau }_{ext}^{a} \right\|^{2}}{2\mu_1}, \\ 
\end{aligned}
\end{equation}
with $\gamma_1$, $\gamma_2$, $\gamma_3$, and $\gamma_4$ are defined based on control parameters in the following formulation.
\begin{equation}
\label{eq20}
\left\{ \begin{aligned}
  & {{\gamma }_{1}}=\frac{\vartheta_1 \left( 2s_{\min}{{l}_{\min }}\left( {{\mathbf{K}}_{1}} \right)-\mu_1 \right)}{{{l}_{\max }}\left( {{\mathbf{M}}^{a}}\left( {{\mathbf{q}}^{a}} \right) \right)} \\ 
 & {{\gamma }_{2}}={{\gamma }_{3}}=\sqrt{\frac{\left( 1-\vartheta_1  \right)\left( 2s_{\min}{{l}_{\min }}\left( {{\mathbf{K}}_{1}} \right)-\mu_1 \right)}{{{l}_{\max }}\left( {{\mathbf{M}}^{a}}\left( {{\mathbf{q}}^{a}} \right) \right)}} \\ 
 & {{\gamma }_{4}}=\frac{{{2}^{\frac{{{\zeta }_{1}}+{{\zeta }_{2}}}{2{{\zeta }_{2}}}}}s_{\min}{{l}_{\min }}\left( {{\mathbf{K}}_{2}} \right)}{{{\left( {{l}_{\max }}\left( {{\mathbf{M}}^{a}}\left( {{\mathbf{q}}^{a}} \right) \right) \right)}^{\frac{{{\zeta }_{1}}+{{\zeta }_{2}}}{2{{\zeta }_{2}}}}}} \\ 
\end{aligned} \right.
\end{equation}
in which $\vartheta_1 \in \left( 0,1 \right)$. Meanwhile, by designing $\gamma_2$ such that ${{\gamma }_{2}}{{k}^{a}}\left( {{\mathbf{q}}^{a}},{{{\mathbf{\dot{q}}}}^{a}} \right)\ge \gamma$, where $\gamma>0$, and using the inequality of norm $\left\| \pmb{\tau }_{ext}^{a} \right\|\le \sqrt{n}{{\left\| \pmb{\tau }_{ext}^{a} \right\|}_{\infty }},\forall \pmb{\tau }_{ext}^{a}\in {{\mathbb{R}}^{n+1}}$, \eqref{eq19} can be rewritten as:
\begin{equation}
\label{eq21}
\begin{aligned}
  & {{{\dot{V}}}_{1}}\,\le -{{\gamma }_{1}}\sqrt{2{{V}_{1}}}{{k}^{a}}\left( {{\mathbf{q}}^{a}},{{{\mathbf{\dot{q}}}}^{a}} \right)-{{\gamma }_{4}}\sqrt{2{{V}_{1}}}{{\left( {{k}^{a}}\left( {{\mathbf{q}}^{a}},{{{\mathbf{\dot{q}}}}^{a}} \right) \right)}^{\frac{{{\zeta }_{1}}+{{\zeta }_{2}}}{2{{\zeta }_{2}}}}} \\ 
 & \,\,\,\,\,\,\,\,\,\,\,\,\,\,\,-{{\gamma }_{3}}\gamma\sqrt{2{{V}_{1}}}+\sqrt{2{{V}_{1}}}\frac{{n}\left\| \pmb{\tau }_{ext}^{a} \right\|_{\infty }^{2}}{2\mu_1}. 
\end{aligned}
\end{equation}
By designing the parameter $\gamma_3$ such that $\gamma_3\gamma\ge\frac{{n}\left\| \pmb{\tau }_{ext}^{a} \right\|_{\infty }^{2}}{2\mu_1}$, we obtain: 
\begin{equation}
\label{eq22}
{{\dot{V}}_{1}}\,\le -{{\gamma }_{1}}\sqrt{2{{V}_{1}}}{{k}^{a}}\left( {{\mathbf{q}}^{a}},{{{\mathbf{\dot{q}}}}^{a}} \right)-{{\gamma }_{4}}\sqrt{2{{V}_{1}}}{{\left( {{k}^{a}}\left( {{\mathbf{q}}^{a}},{{{\mathbf{\dot{q}}}}^{a}} \right) \right)}^{\frac{{{\zeta }_{1}}+{{\zeta }_{2}}}{2{{\zeta }_{2}}}}}.
\end{equation}

For ${{k}^{a}}\left( {{\mathbf{q}}^{a}},{{{\mathbf{\dot{q}}}}^{a}} \right)>k_{d}^{a}+{{\delta }_{3}}$, $\dot{V}_1$ in \eqref{eq22} is negative definite with $\gamma_1,\gamma_4>0$. Then, the system energy ${{k}^{a}}\left( {{\mathbf{q}}^{a}},{{{\mathbf{\dot{q}}}}^{a}} \right)$ converges to $k_{d}^{a}+{{\delta }_{3}}$ in finite time. To calculate the convergence time, \eqref{eq22} is rewritten by using the definition of ${{V}_{1}}\left( e_{E}^{\delta } \right)$.

\begin{equation}
\label{eq23}
\begin{aligned}
  & {{{\dot{V}}}_{1}}\,\le -{{\gamma }_{1}}\sqrt{2{{V}_{1}}}\left( \sqrt{2{{V}_{1}}}+k_{d}^{a}+\delta_3 \right) \\ 
 & \,\,\,\,\,\,\,\,\,\,\,\,\,\,-{{\gamma }_{4}}\sqrt{2{{V}_{1}}}{{\left( \sqrt{2{{V}_{1}}}+k_{d}^{a}+\delta_3 \right)}^{\frac{{{\zeta }_{1}}+{{\zeta }_{2}}}{2{{\zeta }_{2}}}}} \\ 
\end{aligned}
\end{equation}
After the variable transformation from $d{{V}_{1}}$ to $d\left( {{\gamma }_{1}}{{\left( \sqrt{2{{V}_{1}}}+k_{d}^{a}+\delta_3 \right)}^{1-\frac{{{\zeta }_{1}}+{{\zeta }_{2}}}{2{{\zeta }_{2}}}}}+{{\gamma }_{4}} \right)$, it is straightforward to rewrite \eqref{eq23} in the following formulation.
\begin{equation}
\label{eq24}
dt\le -\frac{2{{\zeta }_{2}}}{{{\gamma }_{1}}\left( {{\zeta }_{2}}-{{\zeta }_{1}} \right)}\frac{d\left( {{\gamma }_{1}}{{\left( \sqrt{2{{V}_{1}}}+k_{d}^{a}+\delta_3 \right)}^{1-\frac{{{\zeta }_{1}}+{{\zeta }_{2}}}{2{{\zeta }_{2}}}}}+{{\gamma }_{4}} \right)}{{{\gamma }_{1}}{{\left( \sqrt{2{{V}_{1}}}+k_{d}^{a}+\delta_3 \right)}^{1-\frac{{{\zeta }_{1}}+{{\zeta }_{2}}}{2{{\zeta }_{2}}}}}+{{\gamma }_{4}}}.
\end{equation}
Computing the integral of \eqref{eq24} within the interval $\left[ 0,{{T}_{1}} \right]$, the convergence time $T_1$ is obtained as given in \eqref{eq25}.\\
% \begin{equation}
% \label{eq25}
% {{T}_{1}}\le\frac{2{{\zeta }_{2}}}{{{\gamma }_{1}}\left( {{\zeta }_{2}}-{{\zeta }_{1}} \right)}\ln \left( \frac{{{\gamma }_{1}}{{\left( {{k}^{a}}\left( {{\mathbf{q}}^{a}}\left( 0 \right),{{{\mathbf{\dot{q}}}}^{a}}\left( 0 \right) \right) \right)}^{\frac{{{\zeta }_{2}}-{{\zeta }_{1}}}{2{{\zeta }_{2}}}}}+{{\gamma }_{4}}}{{{\gamma }_{1}}{{\left( k_{d}^{a}+\delta_3+B_{1}^{u} \right)}^{\frac{{{\zeta }_{2}}-{{\zeta }_{1}}}{2{{\zeta }_{2}}}}}+{{\gamma }_{4}}} \right).
% \end{equation}

For ${{k}^{a}}\left( {{\mathbf{q}}^{a}},{{{\mathbf{\dot{q}}}}^{a}} \right)<k_{d}^{a}-{{\delta }_{2}}$, by multiplying both sides of \eqref{eq11} by $\left|e_{E}^{\delta }\right|$ and using \eqref{eq17}-\eqref{eq22}, we obtain $\left|e_{E}^{\delta }\right|\frac{d}{dt}{{k}^{a}}\left( {{\mathbf{q}}^{a}}\left( t \right),{{{\mathbf{\dot{q}}}}^{a}}\left( t \right) \right)\ge-\dot{V}_1\ge0$. Due to $\left|e_{E}^{\delta }\right|>0$ when ${{k}^{a}}\left( {{\mathbf{q}}^{a}},{{{\mathbf{\dot{q}}}}^{a}} \right)<k_{d}^{a}-{{\delta }_{2}}$, it follows that $\frac{d}{dt}{{k}^{a}}\left( {{\mathbf{q}}^{a}}\left( t \right),{{{\mathbf{\dot{q}}}}^{a}}\left( t \right) \right)\ge0$. Then, $\dot{V}_1$ in \eqref{eq22} is negative definite if $\gamma_1,\gamma_4,{{k}^{a}}\left( {{\mathbf{q}}^{a}}\left( 0 \right),{{{\mathbf{\dot{q}}}}^{a}}\left( 0 \right) \right) >0$. The system energy ${{k}^{a}}\left( {{\mathbf{q}}^{a}},{{{\mathbf{\dot{q}}}}^{a}} \right)$ converges to $k_{d}^{a}-{{\delta }_{2}}$ in finite time. Utilizing ${{V}_{1}}\left( e_{E}^{\delta } \right)$, \eqref{eq22} becomes
\begin{equation}
\label{eq26}
\begin{aligned}
  & {{{\dot{V}}}_{1}}\,\le -{{\gamma }_{1}}\sqrt{2{{V}_{1}}}\left( k_{d}^{a}-\delta_2-\sqrt{2{{V}_{1}}} \right) \\ 
 & \,\,\,\,\,\,\,\,\,\,\,\,\,\,-{{\gamma }_{4}}\sqrt{2{{V}_{1}}}{{\left( k_{d}^{a}-\delta_2-\sqrt{2{{V}_{1}}} \right)}^{\frac{{{\zeta }_{1}}+{{\zeta }_{2}}}{2{{\zeta }_{2}}}}}. \\ 
\end{aligned}
\end{equation}
Applying the same approach in \eqref{eq24} and integrating over the interval $[0, T_2]$ yields the convergence time $T_2$ as expressed in \eqref{eq27}.

It should be noted that a dead-zone region is defined by ${{k}^{a}}\left( {{\mathbf{q}}^{a}},{{{\mathbf{\dot{q}}}}^{a}} \right)\in \left[ k_{d}^{a}-{{\delta }_{2}},k_{d}^{a}+{{\delta }_{3}} \right]$ in \eqref{eq9}, within which the energy compensation terms ${{\mathbf{S}}_{1}}\left( {{\mathbf{q}}^{a}},{{{\mathbf{\dot{q}}}}^{a}} \right){{{\mathbf{\dot{q}}}}^{a}}$ and ${{\mathbf{S}}_{2}}\left( {{\mathbf{q}}^{a}},{{{\mathbf{\dot{q}}}}^{a}} \right){{\left\lfloor {{{\mathbf{\dot{q}}}}^{a}} \right\rceil}^{\frac{{{\zeta}_{1}}}{{{\zeta}_{2}}}}}$ in \eqref{eq8} remain inactive. 
%As a result, when  the system energy ${{k}^{a}}\left( {{\mathbf{q}}^{a}},{{{\mathbf{\dot{q}}}}^{a}} \right)$ reaches $k_{d}^{a}+{{\delta }_{3}}$, it is able to traverse the dead-zone and approach domain $k_{d}^{a}-{{\delta }_{2}}$ without exceeding this boundary; and vice versa. 
Accordingly, the system energy ${{k}^{a}}\left( {{\mathbf{q}}^{a}},{{{\mathbf{\dot{q}}}}^{a}} \right)$ is allowed to vary within the dead-zone in response to disturbance effects. Thus, the energy level is bounded by the region $B_1$ in \eqref{eq14} after the finite settling time interval ${{T}_{E}}\le \max \left\{ {{T}_{1}},{{T}_{2}} \right\}$ as presented in \eqref{eq25} and \eqref{eq27}. This completes the proof. \hfill $\blacksquare$
\end{theorem}

\begin{remark}
\label{remark2}
Theorem \ref{thm3} investigates the convergence behavior of the energy level when operating conditions deviate from the designed interval $\left[ k_{d}^a-{{\delta }_{2}},k_{d}^a+{{\delta }_{3}} \right]$ due to the initial setting and external disturbance. In this context, the auxiliary control parameter $\gamma_3$ proportional to $\mathbf{K}_1$ in \eqref{eq8} is designed so that $\gamma_3\gamma\ge\frac{{n}\left\| \pmb{\tau }_{ext}^{a} \right\|_{\infty }^{2}}{2\mu_1}$, where $\gamma>0$, to ensure the convergence of ${{k}^{a}}\left( {{\mathbf{q}}^{a}},{{{\mathbf{\dot{q}}}}^{a}} \right)$ to $B_1$. Meanwhile, the parameter $\gamma_2$ must satisfy the following condition ${{\gamma }_{2}}{{k}^{a}}\left( {{\mathbf{q}}^{a}},{{{\mathbf{\dot{q}}}}^{a}} \right)\ge \gamma$, equivalent to ${{\gamma }_{2}}\ge \Delta =\max \left\{ \frac{\gamma}{{{k}^{a}}\left( \mathbf{q}\left( 0 \right),\mathbf{\dot{q}}\left( 0 \right) \right)},\frac{\max \left\{ {{\delta }_{2}},{{\delta }_{3}} \right\}}{k_{d}^{a}-{{\delta }_{2}}} \right\}$ with ${{{k}^{a}}\left( \mathbf{q}\left( 0 \right),\mathbf{\dot{q}}\left( 0 \right) \right)}>0$. Note that ${{{k}^{a}}\left( \mathbf{q}\left( 0 \right),\mathbf{\dot{q}}\left( 0 \right) \right)}>0$ can be readily ensured through the initial energy of the fictitious flywheel. Because of ${{\gamma }_{2}}={{\gamma }_{3}}$ as defined in \eqref{eq20}, $ {{\gamma }_{2}}={{\gamma }_{3}}\ge\max\left\{\frac{{n}\left\| \pmb{\tau }_{ext}^{a} \right\|_{\infty }^{2}}{2\mu_1\gamma},\Delta\right\}$ should be applied. In addition, the convergence rate is also configured for each specific application through the control parameters $\gamma_1$, $\gamma_4$, $\zeta_1$, and $\zeta_2$ in the explicit settling times  \eqref{eq25} and \eqref{eq27}.
\end{remark}

\subsection{Uniform Ultimate Boundedness Analysis}

In this section, the tracking errors of the mechanical system are proven to be UUB via Lyapunov theory. Beyond the results provided by Theorem \ref{thm3} and Remark \ref{remark2}, the subsequent Lemmas \ref{lm3}-\ref{lm7} are critical in facilitating the proof of Theorem \ref{thm4}.

\begin{lemma}
\label{lm3}
According to Theorem \ref{thm3} and Remark \ref{remark2}, the differential of kinetic energy in the case of $t\ge {{T}_{E}}$ satisfies:
\begin{equation}
\label{eq28}
\left| {{{\dot{k}}}^{a}}\left( {{\mathbf{q}}^{a}},{{{\mathbf{\dot{q}}}}^{a}} \right) \right|\,\le {{\left\| \pmb{\tau }_{ext}^{a}\left(t\right) \right\|}}\sqrt{\frac{2\left( k_{d}^{a}+\delta_3 \right)}{{{l}_{\min }}\left( {{\mathbf{M}}^{a}}\left( {{\mathbf{q}}^{a}} \right) \right)}}.
\end{equation}

\textit{Proof:} Using the differential of kinetic energy ${{{\dot{k}}}^{a}}\left( {{\mathbf{q}}^{a}},{{{\mathbf{\dot{q}}}}^{a}} \right)$ in \eqref{eq11} with ${{\mathbf{S}}_{1}}\left( {{\mathbf{q}}^{a}},{{{\mathbf{\dot{q}}}}^{a}} \right)=\mathbf{0}$ and ${{\mathbf{S}}_{2}}\left( {{\mathbf{q}}^{a}},{{{\mathbf{\dot{q}}}}^{a}} \right)=\mathbf{0}$ due to 
$k_{d}^{a}-{{\delta }_{2}}\le{{k}^{a}}\left( {{\mathbf{q}}^{a}},{{{\mathbf{\dot{q}}}}^{a}} \right)\le k_{d}^{a}+{{\delta }_{3}}$ from Theorem \ref{thm3} and Remark \ref{remark2}, it results in $\left| {{{\dot{k}}}^{a}}\left( {{\mathbf{q}}^{a}},\,{{{\mathbf{\dot{q}}}}^{a}} \right) \right|\,\le \left\| {{{\mathbf{\dot{q}}}}^{a\top}} \right\|{{\left\| \pmb{\tau }_{ext}^{a} \right\|}}$. Based on the second equation of \eqref{eq18}, yields:
\begin{equation}
\label{eq29}
\left| {{{\dot{k}}}^{a}}\left( {{\mathbf{q}}^{a}},{{{\mathbf{\dot{q}}}}^{a}} \right)\, \right|\le {{\left\| \pmb{\tau }_{ext}^{a} \right\|}}\sqrt{\frac{2{{k}^{a}}\left( {{\mathbf{q}}^{a}},{{{\mathbf{\dot{q}}}}^{a}} \right)\,}{{{l}_{\min }}\left( {{\mathbf{M}}^{a}}\left( {{\mathbf{q}}^{a}} \right) \right)}}.
\end{equation}
Adopting the boundedness 
$B_1$ in \eqref{eq14} from Theorem \ref{thm3}, \eqref{eq28} holds. This completes the proof.  \hfill $\blacksquare$
\end{lemma}

\begin{lemma}
\label{lm4}
Utilizing the boundedness of ${{{\dot{k}}}^{a}}\left( {{\mathbf{q}}^{a}},{{{\mathbf{\dot{q}}}}^{a}} \right)$ in Lemma \ref{lm3} and ${{{{k}}}^{a}}\left( {{\mathbf{q}}^{a}},{{{\mathbf{\dot{q}}}}^{a}} \right)$ from Theorem \ref{thm3}, the following inequality holds when $t\ge T_E$: 
\begin{equation}
\label{eq30}
\left| \dot{\alpha }\left( t \right) \right|\le {{\left\| \pmb{\tau }_{ext}^{a}\left(t\right) \right\|}}\sqrt{\frac{k_{d}^{a}+\delta_3}{2{{E}^{a}}{{l}_{\min }}\left( {{\mathbf{M}}^{a}}\left( {{\mathbf{q}}^{a}} \right) \right)\left( k_{d}^{a}-\delta_2 \right)}},
\end{equation}
with $\alpha \left( t \right):=\sqrt{\frac{{{k}^{a}}\left( {{\mathbf{q}}^{a}}\left( t \right),\,{{{\mathbf{\dot{q}}}}^{a}}\left( t \right) \right)}{{{E}^{a}}}}\in \left[ \sqrt{\frac{k_{d}^{a}-\delta_2}{{{E}^{a}}}},\sqrt{\frac{k_{d}^{a}+\delta_3}{{{E}^{a}}}} \right]$.

\textit{Proof:} Taking the derivative of $\alpha \left( t \right)$ with respect to time, we obtain $\dot{\alpha }\left( t \right)=\frac{1}{2}{{\dot{k}}^{a}}\left( {{\mathbf{q}}^{a}},\,{{{\mathbf{\dot{q}}}}^{a}} \right)\sqrt{\frac{1}{{{E}^{a}}{{k}^{a}}\left( {{\mathbf{q}}^{a}},\,{{{\mathbf{\dot{q}}}}^{a}} \right)}}$. The inequality in \eqref{eq30} holds by applying the results in \eqref{eq28} and \eqref{eq14} in the case of $t\ge T_E$. Note that the desired energy level $k_{d}^{a}$ and the control parameter $\delta_2$ can be designed to guarantee $k_{d}^{a}-\delta_2>0$ in \eqref{eq9}. This completes the proof.  \hfill $\blacksquare$
\end{lemma}

\begin{lemma}
\label{lm5}
According to the studies \cite{c22,c27}, the following formulations hold:
\begin{equation}
\label{eq31}
{{\mathbf{w}}^{\top}}{{\mathbf{V}}^{a}}\left( {{\mathbf{q}}^{a}},t \right)=0,
\end{equation}
\begin{equation}
\label{eq32}
\mathbf{R}_1\left( {{\mathbf{q}}^{a}},\,\,{{{\mathbf{\dot{q}}}}^{a}},t \right){{\mathbf{\dot{q}}}^{a}}-\alpha \left( t \right)\mathbf{w}\,\,=\mathbf{R}_1\left( {{\mathbf{q}}^{a}},\,\,{{{\mathbf{\dot{q}}}}^{a}},t \right){{\mathbf{e}}_{v}},
\end{equation}
in which ${{\mathbf{e}}_{v}}:={{\mathbf{\dot{q}}}^{a}}-\alpha \left( t \right){{\mathbf{V}}^{a}}\left( {{\mathbf{q}}^{a}},t \right)$.
\end{lemma}

\begin{lemma}
    \label{lm6}
    Using Theorem \ref{thm3}, Lemmas \ref{lm4}, \ref{lm5}, and considering the mechanical system \eqref{eq6} with external disturbance, the robust SPVFC \eqref{eq8}, the time derivative of the following positive function ${{V}_{2}}\left(\mathbf{e}_v \right):=\dfrac{1}{2}{{\mathbf{e}}_{v}}^{\top}{{\mathbf{M}}^{a}}\left( {{\mathbf{q}}^{a}} \right){{\mathbf{e}}_{v}}$ satisfies
    \begin{equation} 
    \label{eq33}
        \begin{aligned}
            {{\dot{V}}_{2}}\le
            &-2\kappa \alpha \left( t \right){{E}^{a}}\beta \left( t \right){{\mathbf{e}}_{v}}^{\top}{{\mathbf{M}}^{a}}\left( {{\mathbf{q}}^{a}} \right){{\mathbf{e}}_{v}}\\
            &+\frac{{{{\dot{k}}}^{a}}\left( {{\mathbf{q}}^{a}},{{{\mathbf{\dot{q}}}}^{a}} \right)}{4{{k}^{a}}\left( {{\mathbf{q}}^{a}},{{{\mathbf{\dot{q}}}}^{a}} \right)}{{\mathbf{e}}_{v}}^{\top}{{\mathbf{M}}^{a}}\left( {{\mathbf{q}}^{a}} \right){{\mathbf{e}}_{v}}+\left\| {{\mathbf{e}}_{v}}^{\top} \right\|\left\| \pmb{\tau }_{ext}^{a} \right\|,
        \end{aligned}
    \end{equation}
with $\beta \left( t \right):=\frac{1}{2}\left( 1+\frac{{{\mathbf{V}}^{a\top}}\left( {{\mathbf{q}}^{a}},t \right){{\mathbf{M}}^{a}}\left( {{\mathbf{q}}^{a}} \right){{{\mathbf{\dot{q}}}}^{a}}}{2\alpha \left( t \right){{E}^{a}}} \right)$, after $t\ge T_E$.

\textit{Proof:}
Differentiating ${{V}_{2}}\left(\mathbf{e}_v \right)$ and using the augmented mechanical system \eqref{eq6}, we obtain:
\begin{equation}
\label{eq34}
\begin{aligned}
  & {{{\dot{V}}}_{2}}=\mathbf{e}_{v}^{\top}{{\mathbf{M}}^{a}}\left( {{\mathbf{q}}^{a}} \right){{{\mathbf{\dot{e}}}}_{v}}+\frac{1}{2}\mathbf{e}_{v}^{\top}{{{\mathbf{\dot{M}}}}^{a}}\left( {{\mathbf{q}}^{a}} \right){{\mathbf{e}}_{v}} \\ 
 & \,\,\,\,\,\,=\frac{1}{2}\mathbf{e}_{v}^{\top}{{{\mathbf{\dot{M}}}}^{a}}\left( {{\mathbf{q}}^{a}} \right){{\mathbf{e}}_{v}}+\mathbf{e}_{v}^{\top}\left( {{\pmb{\tau }}^{a}}+\pmb{\tau }_{ext}^{a} \right) \\
 &\,\,\,\,\,\,\,\,\,\,\,\,-\mathbf{e}_{v}^{\top}{{\mathbf{C}}^{a}}\left( {{\mathbf{q}}^{a}},{{{\mathbf{\dot{q}}}}^{a}} \right)\left({{\mathbf{e}}_{v}}+\alpha \left( t \right){{\mathbf{V}}^{a}}\left( {{\mathbf{q}}^{a}},t \right) \right)\\
 & \,\,\,\,\,\,\,\,\,\,\,\,-\mathbf{e}_{v}^{\top}\mathbf{M}^a\left( {{\mathbf{q}}^{a}} \right)\left( \alpha \left( t \right){{\nabla }_{t}}{{\mathbf{V}}^{a}}\left( {{\mathbf{q}}^{a}},t \right)+\dot{\alpha }\left( t \right){{\mathbf{V}}^{a}}\left( {{\mathbf{q}}^{a}},t \right) \right). \\ 
\end{aligned}
\end{equation}
Substituting the robust SPVFC \eqref{eq8} into \eqref{eq34}, yields:
\begin{equation}
\label{eq35}
\begin{aligned}
  & {{{\dot{V}}}_{2}}=\frac{1}{2}\mathbf{e}_{v}^{\top}\left( {{{\mathbf{\dot{M}}}}^{a}}\left( {{\mathbf{q}}^{a}} \right)-2{{\mathbf{C}}^{a}}\left( {{\mathbf{q}}^{a}},{{{\mathbf{\dot{q}}}}^{a}} \right) \right){{\mathbf{e}}_{v}} \\ 
 & \,\,\,\,\,\,\,\,\,\,\,\,\,-\mathbf{e}_{v}^{\top}{{\mathbf{S}}_{1}}\left( {{\mathbf{q}}^{a}},{{{\mathbf{\dot{q}}}}^{a}} \right){{{\mathbf{\dot{q}}}}^{a}}-\mathbf{e}_{v}^{\top}{{\mathbf{S}}_{2}}\left( {{\mathbf{q}}^{a}},{{{\mathbf{\dot{q}}}}^{a}} \right){{\left\lfloor {{{\mathbf{\dot{q}}}}^{a}} \right\rceil}^{\frac{{{\zeta}_{1}}}{{{\zeta}_{2}}}}} \\ 
 & \,\,\,\,\,\,\,\,\,\,\,\,\,+\mathbf{e}_{v}^{\top}\left( {{\mathbf{R}}_{1}}\left( {{\mathbf{q}}^{a}},{{{\mathbf{\dot{q}}}}^{a}},t \right){{{\mathbf{\dot{q}}}}^{a}}-\alpha \left( t \right)\mathbf{w} \right)+\mathbf{e}_{v}^{\top}\pmb{\tau }_{ext}^{a} \\ 
 & \,\,\,\,\,\,\,\,\,\,\,\,\,+{{\mathbf{\dot{q}}}^{a\top}}{{\mathbf{R}}_{2}}\left( {{\mathbf{q}}^{a}},{{{\mathbf{\dot{q}}}}^{a}},t \right){{\mathbf{\dot{q}}}^{a}}-\dot{\alpha }\left( t \right)\mathbf{e}_{v}^{\top}{{\mathbf{M}}^{a}}\left( {{\mathbf{q}}^{a}} \right){{\mathbf{V}}^{a}}\left( {{\mathbf{q}}^{a}},t \right) \\ 
 & \,\,\,\,\,\,\,\,\,\,\,\,\,-\alpha \left( t \right){{\mathbf{V}}^{a\top}}\left( {{\mathbf{q}}^{a}},t \right){{\mathbf{R}}_{2}}\left( {{\mathbf{q}}^{a}},{{{\mathbf{\dot{q}}}}^{a}},t \right){{\mathbf{\dot{q}}}^{a}}.
\end{aligned}
\end{equation}
By applying Lemma \ref{lm5}, $\mathbf{e}_{v}^{\top}\left( {{{\mathbf{\dot{M}}}}^{a}}\left( {{\mathbf{q}}^{a}} \right)-2{{\mathbf{C}}^{a}}\left( {{\mathbf{q}}^{a}},{{{\mathbf{\dot{q}}}}^{a}} \right) \right){{\mathbf{e}}_{v}}=0$, ${{\mathbf{\dot{q}}}^{aT}}{{\mathbf{R}}_{2}}\left( {{\mathbf{q}}^{a}},{{{\mathbf{\dot{q}}}}^{a}},t \right){{\mathbf{\dot{q}}}^{a}}=0$, and the results ${{\mathbf{S}}_{1}}\left( {{\mathbf{q}}^{a}},{{{\mathbf{\dot{q}}}}^{a}} \right)=\mathbf{0}$, ${{\mathbf{S}}_{2}}\left( {{\mathbf{q}}^{a}},{{{\mathbf{\dot{q}}}}^{a}} \right)=\mathbf{0}$  from Theorem \ref{thm3} after $t\ge {{T}_{E}}$, we have:
\begin{equation}
\label{eq36}
\begin{aligned}
  & {{{\dot{V}}}_{2}}=-\alpha \left( t \right){{\mathbf{V}}^{a\top}}\left( {{\mathbf{q}}^{a}},t \right){{\mathbf{R}}_{2}}\left( {{\mathbf{q}}^{a}},{{{\mathbf{\dot{q}}}}^{a}},t \right){{{\mathbf{\dot{q}}}}^{a}}\\ 
 & \,\,\,\,\,\,\,\,\,\,\,\,\,\,-\dot{\alpha }\left( t \right)\mathbf{e}_{v}^{\top}{{\mathbf{M}}^{a}}\left( {{\mathbf{q}}^{a}} \right){{\mathbf{V}}^{a}}\left( {{\mathbf{q}}^{a}},t \right) \\ 
 & \,\,\,\,\,\,\,\,\,\,\,\,\,\,+\mathbf{e}_{v}^{\top}{{\mathbf{R}}_{1}}\left( {{\mathbf{q}}^{a}},{{{\mathbf{\dot{q}}}}^{a}},t \right){{\mathbf{e}}_{v}}+\mathbf{e}_{v}^{\top}\pmb{\tau }_{ext}^{a} . \\ 
\end{aligned}
\end{equation}
According to the definitions of ${{\mathbf{e}}_{v}}$, ${{\mathbf{R}}_{2}}\left( {{\mathbf{q}}^{a}},{{{\mathbf{\dot{q}}}}^{a}},t \right)$, and $\mathbf{e}_{v}^{\top}{{\mathbf{R}}_{1}}\left( {{\mathbf{q}}^{a}},{{{\mathbf{\dot{q}}}}^{a}},t \right){{\mathbf{e}}_{v}}=0$, \eqref{eq36} can be rewritten as:
\begin{equation}
\label{eq37}
\begin{aligned}
  & {{{\dot{V}}}_{2}}=-\alpha \left( t \right)\kappa {{\mathbf{V}}^{a\top}}\left( {{\mathbf{q}}^{a}},t \right)\left( \mathbf{P}{{\mathbf{p}}^{\top}}-\mathbf{p}{{\mathbf{P}}^{\top}} \right){{{\mathbf{\dot{q}}}}^{a}}+\mathbf{e}_{v}^{\top}\pmb{\tau }_{ext}^{a}\\ 
 & \,\,\,\,\,\,\,\,\,\,\,\,\,-\dot{\alpha }\left( t \right){{\left( {{{\mathbf{\dot{q}}}}^{a}}-\alpha \left( t \right){{\mathbf{V}}^{a}}\left( {{\mathbf{q}}^{a}},t \right) \right)}^{\top}}{{\mathbf{M}}^{a}}\left( {{\mathbf{q}}^{a}} \right){{\mathbf{V}}^{a}}\left( {{\mathbf{q}}^{a}},t \right). \\ 
\end{aligned}
\end{equation}
Substituting $E^a$, ${{k}^{a}}\left( {{\mathbf{q}}^{a}},{{{\mathbf{\dot{q}}}}^{a}} \right)$, and $\alpha \left( t \right)$ into \eqref{eq37}, yields
\begin{equation}
\label{eq38}
\begin{aligned}
  & {{{\dot{V}}}_{2}}=-\alpha \left( t \right)\kappa \left( \begin{aligned}
  & 4{{\alpha }^{2}}\left( t \right){{\left( {{E}^{a}} \right)}^{2}} \\ 
 & \,\,\,\,\,\,\,\,\,\,\,-{{\left( {{\mathbf{V}}^{a\top}}\left( {{\mathbf{q}}^{a}},t \right){{\mathbf{M}}^{a}}\left( {{\mathbf{q}}^{a}} \right){{{\mathbf{\dot{q}}}}^{a}} \right)}^{2}} \\ 
\end{aligned} \right)\\ 
 & \,\,\,\,\,\,\,\,\,\,\,\,\,-\dot{\alpha }\left( t \right){{{\mathbf{\dot{q}}}}^{a\top}}{{\mathbf{M}}^{a}}\left( {{\mathbf{q}}^{a}} \right){{\mathbf{V}}^{a}}\left( {{\mathbf{q}}^{a}},t \right)\\ 
 & \,\,\,\,\,\,\,\,\,\,\,\,\,+2\dot{\alpha }\left( t \right)\alpha \left( t \right){{E}^{a}} +\mathbf{e}_{v}^{\top}\pmb{\tau }_{ext}^{a}. \\ 
\end{aligned}
\end{equation}
Because the formulation $\dfrac{1}{2}\mathbf{e}_{v}^{\top}{{\mathbf{M}}^{a}}\left( {{\mathbf{q}}^{a}} \right){{\mathbf{e}}_{v}}=2{{\alpha }^{2}}\left( t \right){{E}^{a}}-\alpha \left( t \right){{\mathbf{V}}^{a\top}}\left( {{\mathbf{q}}^{a}},t \right){{\mathbf{M}}^{a}}\left( {{\mathbf{q}}^{a}} \right){{\mathbf{\dot{q}}}^{a}}$ holds; we obtain:
\begin{equation}
\label{eq39}
\begin{aligned}
  & {{{\dot{V}}}_{2}}=-2\kappa \alpha \left( t \right){{E}^{a}}\beta \left( t \right)\mathbf{e}_{v}^{\top}{{\mathbf{M}}^{a}}\left( {{\mathbf{q}}^{a}} \right){{\mathbf{e}}_{v}}\\
&\,\,\,\,\,\,\,\,\,\,\,\,\,+\frac{\dot{\alpha }\left( t \right)}{2\alpha \left( t \right)}\mathbf{e}_{v}^{\top}\mathbf{M}\left( {{\mathbf{q}}^{a}} \right){{\mathbf{e}}_{v}}+\mathbf{e}_{v}^{\top}\pmb{\tau }_{ext}^{a}\\ 
 & \,\,\,\,\,\,\,\le -2\kappa \alpha \left( t \right){{E}^{a}}\beta \left( t \right)\mathbf{e}_{v}^{\top}{{\mathbf{M}}^{a}}\left( {{\mathbf{q}}^{a}} \right){{\mathbf{e}}_{v}}\\ 
 & \,\,\,\,\,\,\,\,\,\,\,\,\,+\frac{ \dot{\alpha }\left( t \right) }{2 \alpha \left( t \right) }\mathbf{e}_{v}^{\top}\mathbf{M}\left( {{\mathbf{q}}^{a}} \right){{\mathbf{e}}_{v}}+\left\| \mathbf{e}_{v}^{\top} \right\|{{\left\| \pmb{\tau }_{ext}^{a} \right\|}}. \\ 
\end{aligned}
\end{equation}
By substituting the definition of $\alpha \left( t \right)$ and $\dot{\alpha }\left( t \right)$ from Lemma \ref{lm4} to \eqref{eq39}, Lemma \ref{lm6} holds. This completes the proof.  \hfill $\blacksquare$
\end{lemma}

In Lemma \ref{lm6}, the first term of \eqref{eq33} contains a time-varying variable  $\beta \left( t \right)=\frac{1}{2}\left( 1+\frac{{{\mathbf{V}}^{a\top}}\left( {{\mathbf{q}}^{a}},t \right){{\mathbf{M}}^{a}}\left( {{\mathbf{q}}^{a}} \right){{{\mathbf{\dot{q}}}}^{a}}}{2\alpha \left( t \right){{E}^{a}}} \right)\in\left[0,1\right]$ due to $\left|{{{\mathbf{V}}^{a\top}}\left( {{\mathbf{q}}^{a}},t \right){{\mathbf{M}}^{a}}\left( {{\mathbf{q}}^{a}} \right){{{\mathbf{\dot{q}}}}^{a}}}\right|\le 2\left|(\alpha\left(t\right)E^a\right|$ by Schwartz’s inequality. Therefore, it is not mathematically rigorous to directly conclude the UUB of $\mathbf{e}_v$. Because of $\dfrac{1}{2}\mathbf{e}_{v}^{\top}{{\mathbf{M}}^{a}}\left( {{\mathbf{q}}^{a}} \right){{\mathbf{e}}_{v}}=2{{\alpha }^{2}}\left( t \right){{E}^{a}}-\alpha \left( t \right){{\mathbf{V}}^{a\top}}\left( {{\mathbf{q}}^{a}},t \right){{\mathbf{M}}^{a}}\left( {{\mathbf{q}}^{a}} \right){{\mathbf{\dot{q}}}^{a}}$, it follows that  $\beta \left( t \right)=1-\frac{\mathbf{e}_{v}^{\top}{{\mathbf{M}}^{a}}\left( {{\mathbf{q}}^{a}} \right){{\mathbf{e}}_{v}}}{8{{k}^{a}}\left( {{{\mathbf{\dot{q}}}}^{a}},{{\mathbf{q}}^{a}} \right)}$. To rigorously derive the conclusion in Theorem \ref{thm4}, the following lemma further investigates the condition of $\beta\left(t\right)$ via $\frac{\mathbf{e}_{v}^{\top}{{\mathbf{M}}^{a}}\left( {{\mathbf{q}}^{a}} \right){{\mathbf{e}}_{v}}}{8{{k}^{a}}\left( {{{\mathbf{\dot{q}}}}^{a}},{{\mathbf{q}}^{a}} \right)}$.

\begin{lemma}
    \label{lm7}
In accordance with Theorem \ref{thm3}, Lemmas \ref{lm3}, \ref{lm6}, and under the given condition ${{\bar{V}}_{3}}\in \left( 0,1 \right)$, the continuously differentiable positive-definite function ${{V}_{3}}\left(\mathbf{e}_v \right):=\dfrac{{{\mathbf{e}}_{v}}^{\top}{{\mathbf{M}}^{a}}\left( {{\mathbf{q}}^{a}} \right){{\mathbf{e}}_{v}}}{8{{k}^{a}}\left( {{\mathbf{q}}^{a}},{{{\mathbf{\dot{q}}}}^{a}} \right)}$ satisfies ${{V}_{3}\left( e_v \right)}\le {{\bar{V}}_{3}}$ for all $t\ge 0$ by designing the control parameter $\kappa$ and the initial condition ${{V}_{3}}\left(\mathbf{e}_v\left(0\right) \right)\le {{\bar{V}}_{3}}$ holds.

\textit{Proof:} According to Lemma \ref{lm6}, we obtain ${{V}_{3}}\left(\mathbf{e}_v \right)=\dfrac{{{V}_{2}\left(\mathbf{e}_v\right)}}{4{{k}^{a}}\left( {{\mathbf{q}}^{a}},{{{\mathbf{\dot{q}}}}^{a}} \right)}$. Conducting its time derivative, yields
    \begin{equation}\label{eq40}
        {{\dot{V}}_{3}}=\frac{1}{4}\frac{{{{\dot{V}}}_{2}}}{{{k}^{a}}\left( {{\mathbf{q}}^{a}},{{{\mathbf{\dot{q}}}}^{a}} \right)}-\frac{1}{4}\frac{{{V}_{2}}{{{\dot{k}}}^{a}}\left( {{\mathbf{q}}^{a}},{{{\mathbf{\dot{q}}}}^{a}} \right)}{{{\left( {{k}^{a}}\left( {{\mathbf{q}}^{a}},{{{\mathbf{\dot{q}}}}^{a}} \right) \right)}^{2}}}.
    \end{equation}
Substituting \eqref{eq33} in Lemma \ref{lm6} into \eqref{eq40}, we have
    \begin{equation} \label{eq41}
        \begin{aligned}
 {{{\dot{V}}}_{3}} \le & -\frac{1}{2{{k}^{a}}\left( {{\mathbf{q}}^{a}},{{{\mathbf{\dot{q}}}}^{a}} \right)}\kappa \alpha \left( t \right){{E}^{a}}\beta \left( t \right){{\mathbf{e}}_{v}}^{\top}{{\mathbf{M}}^{a}}\left( {{\mathbf{q}}^{a}} \right){{\mathbf{e}}_{v}} \\ 
 & -\frac{1}{8}\frac{{{V}_{2}}{{{\dot{k}}}^{a}}\left( {{\mathbf{q}}^{a}},{{{\mathbf{\dot{q}}}}^{a}} \right)}{{{\left( {{k}^{a}}\left( {{\mathbf{q}}^{a}},{{{\mathbf{\dot{q}}}}^{a}} \right) \right)}^{2}}}+\frac{\left\| {{\mathbf{e}}_{v}}^{\top} \right\|\left\| \pmb{\tau }_{ext}^{a} \right\|}{4{{k}^{a}}\left( {{\mathbf{q}}^{a}},{{{\mathbf{\dot{q}}}}^{a}} \right)} \\ 
 \le &  -\left( 4\kappa \alpha \left( t \right){{E}^{a}}\left( 1-{{V}_{3}} \right)+\frac{{{{\dot{k}}}^{a}}\left( {{\mathbf{q}}^{a}},{{{\mathbf{\dot{q}}}}^{a}} \right)}{2{{k}^{a}}\left( {{\mathbf{q}}^{a}},{{{\mathbf{\dot{q}}}}^{a}} \right)} \right){{V}_{3}}\\
 &+\frac{\left\| {{\mathbf{e}}_{v}}^{\top} \right\|\left\| \pmb{\tau }_{ext}^{a} \right\|}{4{{k}^{a}}\left( {{\mathbf{q}}^{a}},{{{\mathbf{\dot{q}}}}^{a}} \right)}.  
\end{aligned}
    \end{equation}
Applying the inequalities in Lemma \ref{lm1} for ${\left\| {{\mathbf{e}}_{v}}^{\top} \right\|\left\| \pmb{\tau }_{ext}^{a} \right\|}$ with an arbitrary positive constant $\mu_2$, \eqref{eq18}, $\left\| \pmb{\tau }_{ext}^{a} \right\|\le \sqrt{n}{{\left\| \pmb{\tau }_{ext}^{a} \right\|}_{\infty }},\forall \pmb{\tau }_{ext}^{a}\in {{\mathbb{R}}^{n+1}}$, and the results from Theorem \ref{thm3}, Lemma \ref{lm3}, we obtain 
    \begin{equation} \label{eq42}
        \begin{aligned}
   &{{{\dot{V}}}_{3}}\le-\left( 4\kappa \alpha \left( t \right){{E}^{a}}\left( 1-{{V}_{3}} \right)+\frac{{{{\dot{k}}}^{a}}\left( {{\mathbf{q}}^{a}},{{{\mathbf{\dot{q}}}}^{a}} \right)}{2{{k}^{a}}\left( {{\mathbf{q}}^{a}},{{{\mathbf{\dot{q}}}}^{a}} \right)} \right){{V}_{3}} \\ 
 & \,\,\,\,\,\,\,\,\,\,\,\,\,\,+\frac{\frac{{{\mu }_{2}}}{2}{{\left\| {{\mathbf{e}}_{v}}^{\top} \right\|}^{2}}+\frac{1}{2{{\mu }_{2}}}{{\left\| \pmb{\tau }_{ext}^{a} \right\|}^{2}}}{4{{k}^{a}}\left( {{\mathbf{q}}^{a}},{{{\mathbf{\dot{q}}}}^{a}} \right)} \\ 
 & \le  -\left( \begin{aligned}
  & 4\kappa \alpha \left( t \right){{E}^{a}}\left( 1-{{V}_{3}} \right)-\frac{{{\mu }_{2}}}{{{l}_{\min }}\left( {{\mathbf{M}}^{a}}\left( {{\mathbf{q}}^{a}} \right) \right)} \\ 
 & -\sqrt{n}{{\left\| \pmb{\tau }_{ext}^{a} \right\|}_{\infty }}\sqrt{\frac{ k_{d}^{a}+\delta_3}{2{{l}_{\min }}\left( {{\mathbf{M}}^{a}}\left( {{\mathbf{q}}^{a}} \right) \right){{\left( k_{d}^{a}-\delta_2 \right)}^{2}}}} \\ 
\end{aligned} \right){{V}_{3}} \\ 
 &\,\,\,\,\,\,\, +\frac{n\left\| \pmb{\tau }_{ext}^{a} \right\|_{\infty }^{2}}{8{{\mu }_{2}}\left( k_{d}^{a}-\delta_2 \right)}.
\end{aligned}
    \end{equation}
    
By designing the control parameters, satisfies
    \begin{equation} \label{eq43}
        \begin{aligned}
            \kappa >&\frac{1}{4\alpha \left( t \right){{E}^{a}}\left( 1-{{{\bar{V}}}_{3}} \right)}\\
            &\times\left( \begin{aligned}
  & \sqrt{n}{{\left\| \pmb{\tau }_{ext}^{a} \right\|}_{\infty }}\sqrt{\frac{ k_{d}^{a}+\delta_3}{2{{l}_{\min }}\left( {{\mathbf{M}}^{a}}\left( {{\mathbf{q}}^{a}} \right) \right){{\left( k_{d}^{a}-\delta_2 \right)}^{2}}}}+{{\kappa }^{*}} \\ 
 & +\frac{{{\mu }_{2}}}{{{l}_{\min }}\left( {{\mathbf{M}}^{a}}\left( {{\mathbf{q}}^{a}} \right) \right)}+\frac{n\left\| \pmb{\tau }_{ext}^{a} \right\|_{\infty }^{2}}{8{{\mu }_{2}}\left( k_{d}^{a}-\delta_2 \right){{{\bar{V}}}_{3}}} \\ 
\end{aligned} \right)
        \end{aligned}
    \end{equation}
with $\kappa^*$ is positive constant, and ${{V}_{3}}\left( \mathbf{e}_v\left({{t}^{*}\ge0} \right)\right)\le {{\bar{V}}_{3}}$, yields
    \begin{equation} \label{eq44}
        \begin{aligned}
  {{{\dot{V}}}_{3}\left(\mathbf{e}_v\left(t^*\right)\right)} \le -\left( \frac{1-{{V}_{3}}}{1-{{{\bar{V}}}_{3}}}-1 \right)\\
  &  \hspace{-4cm}\times \left( \begin{aligned}
  & \sqrt{n}{{\left\| \pmb{\tau }_{ext}^{a} \right\|}_{\infty }}\sqrt{\frac{ k_{d}^{a}+\delta_3}{2{{l}_{\min }}\left( {{\mathbf{M}}^{a}}\left( {{\mathbf{q}}^{a}} \right) \right){{\left( k_{d}^{a}-\delta_2 \right)}^{2}}}} \\ 
 & +\frac{{{\mu }_{2}}}{{{l}_{\min }}\left( {{\mathbf{M}}^{a}}\left( {{\mathbf{q}}^{a}} \right)\right)} \\ 
\end{aligned} \right){{V}_{3}} \\ 
 &\hspace{-4cm} -\frac{1-{{V}_{3}}}{1-{{{\bar{V}}}_{3}}}{{\kappa }^{*}}{{V}_{3}}-\left( \frac{1-{{V}_{3}}}{1-{{{\bar{V}}}_{3}}}\frac{{{V}_{3}}}{{{{\bar{V}}}_{3}}}-1 \right)\frac{n\left\| \pmb{\tau }_{ext}^{a} \right\|_{\infty }^{2}}{8{{\mu }_{2}}\left( k_{d}^{a}-\delta_2 \right)}  
\end{aligned}
    \end{equation}
Obviously, ${{{\dot{V}}}_{3}\left(\mathbf{e}_v\left(t^*\right)\right)}$ is non-positive function. Thus, after the time instant $t\ge{{t}^{*}}$, the system does not reveal any tendency to escape the constrained region defined by ${{V}_{3}}\left(\mathbf{e}_v \right)\le {{\bar{V}}_{3}}$. In this manner, if ${{V}_{3}}\left(\mathbf{e}_v\left(0\right) \right)\le {{\bar{V}}_{3}}$, ${{V}_{3}}\left(\mathbf{e}_v \right)\le {{\bar{V}}_{3}}$ for all $t\ge0$. This completes the proof.  \hfill $\blacksquare$
\end{lemma}

\begin{theorem}
\label{thm4}
Under Theorem \ref{thm3}, Lemmas \ref{lm4}, \ref{lm6}, and \ref{lm7}, the position and velocity errors of the closed-loop mechanical system, consisting of the augmented system with external disturbance \eqref{eq6}, the time-varying velocity field \eqref{eq7}, and the robust SPVFC \eqref{eq8}, converge to a UUB region $B_2$ as
\begin{equation}
\label{eq45}
{{B}_{2}}:=\left\{ {{\mathbf{e}}_{s}}\in {{\mathbb{R}}^{2n+1 }}\left| \left\| {{\mathbf{e}}_{s}} \right\|\le B_{2}^{u} \right. \right\},
\end{equation}
where $B_{2}^{u}:=\sqrt{\frac{\Lambda}{{{\left( 1-{{\vartheta }_{2}} \right)}}\underline{\Upsilon } \Xi }}$, ${{\mathbf{e}}_{s}}:={{\left[ \begin{matrix}
   \mathbf{e}_{p}^{\top} & \mathbf{e}_{v}^{\top}  \\
\end{matrix} \right]}^{\top}}$ with ${{\mathbf{e}}_{p}}:=\mathbf{q}\left( t \right)-\mathbf{q}_d\left( t \right)$ and  ${{\mathbf{e}}_{v}}:={{\mathbf{\dot{q}}}^{a}}\left( t \right)-\alpha \left( t \right){{\mathbf{V}}^{a}}\left( {{\mathbf{q}}^{a}},t \right)$ after the settling time
\begin{equation} \label{eq51}
    {{T}_{s}}\le \frac{{{V}_{4}}\left( t=0 \right)-\underline{{\Upsilon } }{{\left( B_{2}^{u} \right)}^{2}}}{{{\vartheta }_{2}}\Xi \underline{\Upsilon }{{\left( B_{2}^{u} \right)}^{2}} }.
\end{equation}
 Here, $\Lambda$ is a function of ${{\left\| {{{\mathbf{\dot{q}}}}_{d}} \right\|}_{\infty }}$ and ${{\left\| \pmb{\tau }_{ext}^{a} \right\|}_{\infty }}$ and the auxiliary parameter $\Xi>0$ is proportional to $\pmb{\psi}$ in \eqref{eq7}, $\kappa$ in \eqref{eq8}, $E^a$, while $\vartheta_2 \in \left( 0,1 \right)$ and $\underline{\Upsilon } =\min \left\{ \frac{1}{2},\frac{{{l}_{\min }}\left( {{\mathbf{M}}^{a}}\left( {{\mathbf{q}}^{a}} \right) \right)}{2} \right\}$.

\textit{Proof:} Considering a Lyapunov function as
\begin{equation}
\label{eq46}
{{V}_{4}}\left( \mathbf{e}_p,\mathbf{e}_v \right):=\frac{1}{2}\mathbf{e}_{p}^{\top}{{\mathbf{e}}_{p}}+{{V}_{2}}\left(\mathbf{e}_v \right).
\end{equation}
Differentiating \eqref{eq46} and using the time-varying velocity field \eqref{eq7} and \eqref{eq33} in Lemma \ref{lm6}, we obtain:
\begin{equation}
\label{eq47}
\begin{aligned}
{{{\dot{V}}}_{4}}\,\le   & -\alpha \left( t \right)\mathbf{e}_{p}^{\top}\pmb{\psi }{{\mathbf{e}}_{p}}+ \left| \alpha \left( t \right)-1 \right|\left\| \mathbf{e}_{p}^{\top} \right\|\left\| {{{\mathbf{\dot{q}}}}_{d}} \right\| \\ 
 &-2\kappa \alpha \left( t \right){{E}^{a}}\beta \left( t \right){{\mathbf{e}}_{v}}^{\top}{{\mathbf{M}}^{a}}\left( {{\mathbf{q}}^{a}} \right){{\mathbf{e}}_{v}}\\
            &+\frac{{{{\dot{k}}}^{a}}\left( {{\mathbf{q}}^{a}},{{{\mathbf{\dot{q}}}}^{a}} \right)}{4{{k}^{a}}\left( {{\mathbf{q}}^{a}},{{{\mathbf{\dot{q}}}}^{a}} \right)}{{\mathbf{e}}_{v}}^{\top}{{\mathbf{M}}^{a}}\left( {{\mathbf{q}}^{a}} \right){{\mathbf{e}}_{v}}+\left\| {{\mathbf{e}}_{v}}^{\top} \right\|\left\| \pmb{\tau }_{ext}^{a} \right\|. \\ 
\end{aligned}
\end{equation}
By using the results in Lemma \ref{lm3}, Theorem \ref{thm3}, and applying the inequalities in Lemma \ref{lm1} for $\left\| {{\mathbf{e}}_{v}}^{\top} \right\|\left\| \pmb{\tau }_{ext}^{a} \right\|$ with an arbitrary positive constant $\mu_2$, we have
\begin{equation}
\label{eq48}
\begin{aligned}
  & {{{\dot{V}}}_{4}}\,\le -\alpha \left( t \right){{l}_{\min }}\left( \pmb{\psi } \right){{{\left\| {{\mathbf{e}}_{p}} \right\|}^{2}}}+ \frac{{{\left\| {{\mathbf{e}}_{p}} \right\|}^{2}}}{2}+\frac{{{\left( \alpha \left( t \right)-1 \right)}^{2}}{{\left\| {{{\mathbf{\dot{q}}}}_{d}} \right\|}^{2}}}{2}\\ 
  & \,\,\,\,\,\,\,\,\,\,\,\,\,\,-2\kappa \alpha \left( t \right){{E}^{a}}\beta \left( t \right)\mathbf{e}_{v}^{\top}{{\mathbf{M}}^{a}}\left( {{\mathbf{q}}^{a}} \right){{\mathbf{e}}_{v}} \\
 & \,\,\,\,\,\,\,\,\,\,\,\,\,\,+\frac{{{\mu_2\left\| {{\mathbf{e}}_{v}} \right\|}^{2}}}{2}+\frac{\left\| \pmb{\tau }_{ext}^{a} \right\|^{2}}{2\mu_2}\\
 &\,\,\,\,\,\,\,\,\,\,\,\,\,\,+{{{\frac{\left\| \pmb{\tau }_{ext}^{a} \right\|}{2}}}\sqrt{\frac{ k_{d}^{a}+\delta_3}{2{{l}_{\min }}\left( {{\mathbf{M}}^{a}}\left( {{\mathbf{q}}^{a}} \right) \right){{\left( k_{d}^{a}-\delta_2 \right)}^{2}}}}}\mathbf{e}_{v}^{\top}\mathbf{M}\left( {{\mathbf{q}}^{a}} \right){{\mathbf{e}}_{v}} \\  
 & \,\,\,\,\,\,\,\le -\chi_1\frac{{{\left\| {{\mathbf{e}}_{p}} \right\|}^{2}}}{2}+\frac{{{\left( \alpha \left( t \right)-1 \right)}^{2}}{{\left\| {{{\mathbf{\dot{q}}}}_{d}} \right\|}^{2}}}{2} \\ 
 & \,\,\,\,\,\,\,\,\,\,\,\,\,\,+\frac{\left\| \pmb{\tau }_{ext}^{a} \right\|^{2}}{2\mu_2}-\chi_2 \frac{\mathbf{e}_{v}^{\top}{{\mathbf{M}}^{a}}\left( {{\mathbf{q}}^{a}} \right){{\mathbf{e}}_{v}}}{2}, \\ 
\end{aligned}
\end{equation}
where the term $\beta \left( t \right)=1-\frac{\mathbf{e}_{v}^{\top}{{\mathbf{M}}^{a}}\left( {{\mathbf{q}}^{a}} \right){{\mathbf{e}}_{v}}}{8{{k}^{a}}\left( {{{\mathbf{\dot{q}}}}^{a}},{{\mathbf{q}}^{a}} \right)}=1-V_3\ge {1-{{\bar{V}}_{3}}}>0$ by choosing $1-\beta \left( t=0 \right)\le {{\bar{V}}_{3}}$, as mentioned in Lemma \ref{lm7}. Furthermore, the control parameters are designed such as ${{\chi }_{1}}=2\alpha \left( t \right){{l}_{\min }}\left( \pmb{\psi } \right)-1>0$ and $\chi_2 =4\kappa \left( {1-{{\bar{V}}_{3}}} \right)\sqrt{{{E}^{a}}\left( k_{d}^{a}-\delta_2 \right)}-\frac{\mu_2}{{{l}_{\min }}\left( {{\mathbf{M}}^{a}}\left( {{\mathbf{q}}^{a}} \right) \right)}-\sqrt{n}{{\left\| \pmb{\tau }_{ext}^{a} \right\|}_{\infty }}\sqrt{\frac{k_{d}^{a}+\delta_3}{2{{l}_{\min }}\left( {{\mathbf{M}}^{a}}\left( {{\mathbf{q}}^{a}} \right) \right){{\left( k_{d}^{a}-\delta_2 \right)}^{2}}}}>0$ with $\left\| \pmb{\tau }_{ext}^{a} \right\|\le \sqrt{n}{{\left\| \pmb{\tau }_{ext}^{a} \right\|}_{\infty }},\forall \pmb{\tau }_{ext}^{a}\in {{\mathbb{R}}^{n+1}}$. By considering the parameter such that $\Xi =\min \left\{ \chi_1,\chi_2  \right\}$ and using the inequality of norm $\left\| {{{\mathbf{\dot{q}}}}_{d}} \right\|\le \sqrt{n}{{\left\| {{{\mathbf{\dot{q}}}}_{d}} \right\|}_{\infty }},\,\forall {{\mathbf{\dot{q}}}_{d}}\in {{\mathbb{R}}^{n}}$, it is straightforward to derive the following formulation.
\begin{equation}
\label{eq49}
\begin{aligned}
  & {{{\dot{V}}}_{4}}\le \,-\Xi {{V}_{4}}+\frac{{{\left( \alpha \left( t \right)-1 \right)}^{2}}n\left\| {{{\mathbf{\dot{q}}}}_{d}} \right\|_{\infty }^{2}}{2}+\frac{{n}\left\| \pmb{\tau }_{ext}^{a} \right\|_{\infty }^{2}}{2\mu_2}  \\ 
 & \,\,\,\,\,\,\,\le -{{\vartheta }_{2}}\Xi {{V}_{4}}-\,\left( 1-{{\vartheta }_{2}} \right)\Xi {{V}_{4}}+\Lambda  \\ 
\end{aligned}
\end{equation}
in which $\vartheta_2 \in \left( 0,1 \right)$ and $\Lambda=\frac{{{\left( \sqrt{\frac{k_{d}^{a}+\delta_3}{{{E}^{a}}}}-1 \right)}^{2}}n\left\| {{{\mathbf{\dot{q}}}}_{d}} \right\|_{\infty }^{2}}{2}+\frac{{n}\left\| \pmb{\tau }_{ext}^{a} \right\|_{\infty }^{2}}{2\mu_2}$ with the bounded term $\alpha \left( t \right)$ in Lemma \ref{lm4}. For ${{V}_{4}}>\frac{\Lambda}{{{\left( 1-{{\vartheta }_{2}} \right)}}\Xi }$, it results in ${{\dot{V}}_{4}}\,\le -{{{\vartheta }_{2}}}\Xi {{V}_{4}}$; therefore, based on Theorem \ref{thm1}, the decrease of the Lyapunov function $V_4$ drives the norm of the errors of the closed-loop system $\left\| {{\mathbf{e}}_{s}} \right\|$ into a bounded region ${B_{2}^{u}}=\sqrt{\frac{\Lambda}{{{\left( 1-{{\vartheta }_{2}} \right)}}\underline{\Upsilon } \Xi }}$ with  $\underline{\Upsilon } =\min \left\{ \frac{1}{2},\frac{{{l}_{\min }}\left( {{\mathbf{M}}^{a}}\left( {{\mathbf{q}}^{a}} \right) \right)}{2} \right\}$.

To determine the convergence time of the position and velocity errors toward the bounded region ${B_{2}^{u}}$, the time derivative of $V_4$ can be rewritten in the following form ${{\dot{V}}_{4}}\le -{{\vartheta }_{2}}\Xi \underline{\Upsilon }{{\left( B_{2}^{u} \right)}^{2}}$ for
${{V}_{4}}>\frac{\Lambda}{{{\left( 1-{{\vartheta }_{2}} \right)}}\Xi }$. By calculating its integral within the interval $\left[ 0,{{T}_{s}} \right]$, we obtain:
\begin{equation} \label{eq50}
\int\limits_{V_4\left(t=0\right)}^{{V_4\left(t={T}_{s}\right)}}{d{{V}_{4}}}\le -\int\limits_{t=0}^{{{T}_{s}}}{{{\vartheta }_{2}}\Xi \underline{\Upsilon }{{\left( B_{2}^{u} \right)}^{2}} dt}
\end{equation}
It is straightforward to drive the settling time $T_s$ in \eqref{eq51} by solving \eqref{eq50}.

Thus, the position and velocity errors converge to the UUB region $B_2$ after the finite settling time interval ${{T}_{s}}$. This completes the proof. \hfill $\blacksquare$
\end{theorem}

\begin{remark}
\label{remark3}
As established in Lemma \ref{lm6} and Theorem \ref{thm4}, the control parameters $\kappa$, $E^a$, and $\pmb{\psi}$ are designed to satisfy the condition in \eqref{eq43} and $\Xi>0$. In addition, these parameters enable adjusting the specification of the bounded region $B_2$ as well as the convergence rate $T_s$ of the system states. For example, larger values of the control parameters $\kappa$, $E^a$, and $\pmb{\psi}$ in $\Xi$ lead to a tighter bounded region $B_2$, thereby enhancing the task performance, while potentially increasing the settling time $T_s$; and vice versa. Thus, parameter selection should reflect application-specific priorities, such as performance requirements and responsiveness in dynamic environments. 

%However, excessively conservative or aggressive parameter choices may compromise control quality. Obviously, from the bounded region in Theorem \ref{thm4}, larger values of $\kappa$, $E^a$, and $\pmb{\psi}$ in $\Xi$ can reduce the bound size, but may induce chattering or actuator saturation. Conversely, selecting smaller values yields smoother responses and lower sensitivity while degrading performance. Furthermore, as mentioned in \eqref{eq51}, the convergence time can be decreased if the bounded region is relaxed; and vice versa. Thus, parameter selection should reflect application-specific priorities, such as responsiveness, safety margins, performance requirements, and actuator limitations.
\end{remark}

\begin{remark}
\label{remark4}
It is worth emphasizing that robust time-varying SPVFC \eqref{eq8} guarantees the convergence of both the energy level and the system states of the closed-loop mechanical system, as established in Theorems \ref{thm3} and \ref{thm4}, respectively. The convergence of the energy level is first investigated in Theorem \ref{thm3} to drive the closed-loop mechanical system toward the safe operating region $B_1$. However, the result from Theorem \ref{thm3} does not indicate whether the system is performing the desired task, as it only reflects the magnitude of motion. As a result, under Theorem \ref{thm3}, the ability of the system to perform the desired task is further assessed through the convergence of state errors in Theorem \ref{thm4}. Furthermore, although both Theorems are derived under the same control law, their results are complementary and do not exhibit conflict. This is because the outcomes in Theorem \ref{thm3} are regulated by the parameters $\mathbf{K_1}$, $\mathbf{K_2}$, ${{\zeta}_{1}}$, ${{\zeta}_{2}}$, ${{\eta }_{\max }}$, ${{\eta }_{\min }}$, ${{\delta }_{1}}$, ${{\delta }_{2}}$, ${{\delta }_{3}}$, ${{\delta }_{4}}$  in ${{\mathbf{S}}_{1}}\left( {{\mathbf{q}}^{a}},{{{\mathbf{\dot{q}}}}^{a}} \right)$ and ${{\mathbf{S}}_{2}}\left( {{\mathbf{q}}^{a}},{{{\mathbf{\dot{q}}}}^{a}} \right)$, whereas Theorem \ref{thm4} relies on different control functions ${{\mathbf{R}}_{1}}\left( {{\mathbf{q}}^{a}},{{{\mathbf{\dot{q}}}}^{a}},t \right)$ and ${{\mathbf{R}}_{2}}\left( {{\mathbf{q}}^{a}},{{{\mathbf{\dot{q}}}}^{a}},t \right)$ with the parameters $\kappa$, $E^a$, and $\pmb{\psi}$. In addition to the conditions established in Theorems \ref{thm3} and 
\ref{thm4}, actuator limitations must also be considered in practical implementations. From the controller in \eqref{eq8}, the control parameters $\mathbf{K_1}$, $\mathbf{K_2}$, ${{\zeta}_{1}}$, ${{\zeta}_{2}}$, ${{\eta }_{\max }}$, ${{\eta }_{\min }}$, $\kappa$, and $E^a$ directly influence the magnitude of the control input. Specifically, selecting excessively aggressive parameter values may improve responsiveness and reduce the bound size, but can also produce large control amplitudes or even actuator saturation. Conversely, smaller parameter values generally lead to smoother control actions and lower control effort, at the expense of slower convergence and degraded tracking performance. Therefore, the controller parameters implicitly determine the practical bounds of the control input and should be selected such that the resulting control signals remain within actuator limitations while guaranteeing the control objectives. 
\end{remark}

\subsection{Power Flow Constraint}

The theoretical evidence above reveals that the proposed approach can relax the conservative nature of PBC in a controlled manner while constraining the energy level of the closed-loop system and performing the desired task under external disturbances. However, high instantaneous power flows may compromise safety during physical contacts. Therefore, the subsection analyzes the power flow behaviors between the closed-loop system and its physical environment.

\begin{corollary}
 \label{coro1}
Under the same conditions as Theorems \ref{thm2}, \ref{thm3} and \ref{thm4}, the power flow between the closed-loop mechanical system and its physical environment is bounded by $\left| P\left( t \right) \right|\le \left\| {{{\mathbf{\dot{q}}}}^{\top}} \right\|{{\left\| \pmb{\tau }_{ext}^{a} \right\|}}\le {{B}_{3}}\sqrt{n}{{\left\| \pmb{\tau }_{ext}^{a} \right\|}_{\infty }}$ after $t\ge T_s$, where ${{B}_{3}}\ge 0$ satisfies $\left\| \mathbf{\dot{q}}\left( t \right) \right\|\le {{B}_{3}}\le \sqrt{3}\left( \left\| {{\mathbf{X}}_{d}} \right\|+B_{2}^{u} \right)$ with ${{\mathbf{X}}_{d}}={{\left[ \begin{matrix}
   \mathbf{q}_{d}^{\top}\left( t \right) & \alpha \left( t \right){{\mathbf{V}}^{\top}}\left( \mathbf{q},t \right) & \alpha \left( t \right){{V}_{f}}\left( \mathbf{q},t \right)  \\
\end{matrix} \right]}^{\top}}$. 

\textit{Proof:} Considering the following power flow formulation and combining with the augmented mechanical system \eqref{eq6} and the robust time-varying SPVFC \eqref{eq8}, we obtain
\begin{equation}
\label{eq52}
\begin{aligned}
  & P\left( t \right)={{{\mathbf{\dot{q}}}}^{\top}}{{\pmb{\tau }}_{ext}}={{{\mathbf{\dot{q}}}}^{a\top}}\pmb{\tau }_{ext}^{a} \\ 
 & \,\,\,\,\,\,\,\,\,\,\,\,\,={{{\dot{k}}}^{a}}\left( {{\mathbf{q}}^{a}},{{{\mathbf{\dot{q}}}}^{a}} \right)+{{{\mathbf{\dot{q}}}}^{a\top}}\,{{\mathbf{S}}_{1}}\left( {{\mathbf{q}}^{a}},{{{\mathbf{\dot{q}}}}^{a}} \right){{{\mathbf{\dot{q}}}}^{a}} \\ 
 & \,\,\,\,\,\,\,\,\,\,\,\,\,\,\,\,\,\,\,+{{{\mathbf{\dot{q}}}}^{a\top}}\,{{\mathbf{S}}_{2}}\left( {{\mathbf{q}}^{a}},{{{\mathbf{\dot{q}}}}^{a}} \right){{\left\lfloor {{{\mathbf{\dot{q}}}}^{a}} \right\rceil}^{\frac{{{\zeta}_{1}}}{{{\zeta}_{2}}}}}. \\ 
\end{aligned}
\end{equation}
in which $P\left( t \right)>0$ expresses the power flow from the physical environment to the mechanical system, $P\left( t \right)<0$ indicates the opposite direction of the power flow. According to Theorem \ref{thm3}, for $t<{{T}_{E}}$, the energy level ${{k}^{a}}\left( {{\mathbf{q}}^{a}},{{{\mathbf{\dot{q}}}}^{a}} \right)$ approaches ${{B}_{1}}$, which corresponds to the negative energy derivative ${{{\dot{k}}}^{a}}\left( {{\mathbf{q}}^{a}},{{{\mathbf{\dot{q}}}}^{a}} \right)$, when ${{k}^{a}}\left( {{\mathbf{q}}^{a}},{{{\mathbf{\dot{q}}}}^{a}} \right)>k_{d}^{a}+{{\delta }_{3}}$. Additionally, the terms ${{{\mathbf{\dot{q}}}}^{a\top}}{{\mathbf{S}}_{1}}\left( {{\mathbf{q}}^{a}},{{{\mathbf{\dot{q}}}}^{a}} \right){{{\mathbf{\dot{q}}}}^{a}}$ and ${{{\mathbf{\dot{q}}}}^{a\top}}{{\mathbf{S}}_{2}}\left( {{\mathbf{q}}^{a}},{{{\mathbf{\dot{q}}}}^{a}} \right){{\left\lfloor {{{\mathbf{\dot{q}}}}^{a}} \right\rceil}^{\frac{{{\zeta}_{1}}}{{{\zeta}_{2}}}}}$ in \eqref{eq52} are the positively defined values in the case of ${{k}^{a}}\left( {{\mathbf{q}}^{a}},{{{\mathbf{\dot{q}}}}^{a}} \right)>k_{d}^{a}+{{\delta }_{3}}$. Meanwhile, for ${{k}^{a}}\left( {{\mathbf{q}}^{a}},{{{\mathbf{\dot{q}}}}^{a}} \right)<k_{d}^{a}-{{\delta }_{2}}$, the differential of kinetic energy is positive, whereas both terms on the right-hand side of \eqref{eq52} are negative. As a result, the power flow $P\left( t \right)$ is gradually mitigated when $t<{{T}_{E}}$. It is worth noting that the compensation of the two terms ${{{\mathbf{\dot{q}}}}^{a\top}}{{\mathbf{S}}_{1}}\left( {{\mathbf{q}}^{a}},{{{\mathbf{\dot{q}}}}^{a}} \right){{{\mathbf{\dot{q}}}}^{a}}$ and ${{{\mathbf{\dot{q}}}}^{a\top}}{{\mathbf{S}}_{2}}\left( {{\mathbf{q}}^{a}},{{{\mathbf{\dot{q}}}}^{a}} \right){{\left\lfloor {{{\mathbf{\dot{q}}}}^{a}} \right\rceil}^{\frac{{{\zeta}_{1}}}{{{\zeta}_{2}}}}}$ is consistently regulated by the energy level through \eqref{eq9}, thereby preventing excessive compensation.

Until $t\ge {{T}_{E}}$, based on Theorem \ref{thm3} and Remark \ref{remark1}, \eqref{eq52} can be rewritten in the following formulation $P\left( t \right)={{{\dot{k}}}^{a}}\left( {{\mathbf{q}}^{a}},{{{\mathbf{\dot{q}}}}^{a}} \right)$ due to $k_{d}^{a}-{{\delta }_{2}}\le{{k}^{a}}\left( {{\mathbf{q}}^{a}},{{{\mathbf{\dot{q}}}}^{a}} \right)\le k_{d}^{a}+{{\delta }_{3}}$. Using the result in \eqref{eq28} from Lemma \ref{lm3}, we obtain:
\begin{equation}
\label{eq53}
\begin{aligned}
  & \left| P\left( t \right) \right|\le\left| {{{\dot{k}}}^{a}}\left( {{\mathbf{q}}^{a}},{{{\mathbf{\dot{q}}}}^{a}} \right) \right|\le \left\| {{{\mathbf{\dot{q}}}}^{a\top}} \right\|{{\left\| \pmb{\tau }_{ext}^{a} \right\|}} \\
  &\,\,\,\,\,\,\,\,\,\,\,\,\,\,\,\,\, \le \sqrt{n}{{\left\| \pmb{\tau }_{ext}^{a} \right\|}_{\infty }}\sqrt{\frac{2\left( k_{d}^{a}+\delta_3 \right)}{{{l}_{\min }}\left( {{\mathbf{M}}^{a}}\left( {{\mathbf{q}}^{a}} \right) \right)}} \\
\end{aligned}
\end{equation}
However, the bounded domain for power flow in \eqref{eq53} is relatively large, attributed to the use of the energy-bounded region of the augmented system, as presented in Theorem \ref{thm3}. According to Theorem \ref{thm4} after $t\ge{T_s}$ and the triangle inequality reverse form, we have
\begin{equation}
\label{eq54}
\left| \left\| \mathbf{X} \right\|-\left\| {{\mathbf{X}}_{d}} \right\| \right|\le \left\| {{\mathbf{e}}_{s}} \right\|\le B_{2}^{u}, \\
\end{equation}
in which $\mathbf{X}={{\left[ \begin{matrix}
   {{\mathbf{q}}^{\top}}\left( t \right) & {{{\mathbf{\dot{q}}}}^{\top}}\left( t \right) & {\dot{q}_{f}}\left( t \right)  \\
\end{matrix} \right]}^{\top}}$ and ${{\mathbf{X}}_{d}}={{\left[ \begin{matrix}
   \mathbf{q}_{d}^{\top}\left( t \right) & \alpha \left( t \right){{\mathbf{V}}^{\top}}\left( \mathbf{q},t \right) & \alpha \left( t \right){{V}_{f}}\left( \mathbf{q},t \right)  \\
\end{matrix} \right]}^{\top}}$. Based on \eqref{eq54} and using the Cauchy-Schwarz’s inequality in \cite{c31}, \eqref{eq54} can be rewritten as
\begin{equation}
\label{eq55}
\left\| \mathbf{q}\left( t \right) \right\|+\left\| \mathbf{\dot{q}}\left( t \right) \right\|+\left| {{{\dot{q}}}_{f}}\left( t \right) \right|\le \sqrt{3}\left\| \mathbf{X} \right\|\le \sqrt{3}\left( \left\| {{\mathbf{X}}_{d}} \right\|+B_{2}^{u} \right). \\
\end{equation}
From \eqref{eq55}, there exists a ${{B}_{3}}\ge 0$ such that $\left\| \mathbf{\dot{q}}\left( t \right) \right\|\le {{B}_{3}}\le \sqrt{3}\left( \left\| {{\mathbf{X}}_{d}} \right\|+B_{2}^{u} \right)$. In this manner, the bounded region of power flow can be presented in the following formulation $\left| P\left( t \right) \right|\le \left\| {{{\mathbf{\dot{q}}}}^{\top}} \right\|{{\left\| \pmb{\tau }_{ext}^{a} \right\|}}\le {{B}_{3}}\sqrt{n}{{\left\| \pmb{\tau }_{ext}^{a} \right\|}_{\infty }}$.  
This completes the proof.  \hfill $\blacksquare$
\end{corollary}

\begin{remark}
\label{remark5}
In comparison with the classical PBC in the context of physical contacts introduced in \cite{c18,c19,c20,c21,c22,c23,c24}, the proposed approach explicitly addresses the conservativeness issue by allowing a controlled relaxation of passivity, as detailed in Section \ref{3b}, rather than enforcing purely passive behaviors. In contrast to the switching passivity method proposed in \cite{c28}, which may cause discontinuities and undesirable oscillations, this work ensures a smooth and continuous transition between passive and non-passive behaviors. Notably, while the existing papers \cite{c25,c26,c27} ignore external perturbations commonly encountered in unstructured environments, the present work rigorously analyzes the convergence of both the energy level and the system states toward the bounded regions $B_1$ and $B_2$ in the presence of external forces, as established in Theorems \ref{thm3} and \ref{thm4}. This ensures reliable task execution under realistic physical interactions. Furthermore, the instantaneous power flow constraint is also analyzed in Corollary \ref{coro1}  to further enhance physical interaction safety.
\end{remark}

% \begin{equation}
% \label{eq111}
% {{\mathbf{M}}^{a}}\left( {{\mathbf{q}}^{a}} \right){{\mathbf{\ddot{q}}}^{a}}+{{\mathbf{C}}^{a}}\left( {{\mathbf{q}}^{a}},{{{\mathbf{\dot{q}}}}^{a}} \right){{\mathbf{\dot{q}}}^{a}}={{\pmb{\tau }}^{a}}+\pmb{\tau }_{ext}^{a},
% \end{equation}

\section{Simulation Example}
Numerical simulation examples of the robot manipulator are given in this section to validate the effectiveness of the proposed method.
\subsection{Simulation Settings}
\label{4a}
\begin{figure}[t]
    \centering
\includegraphics[width=0.8\linewidth]{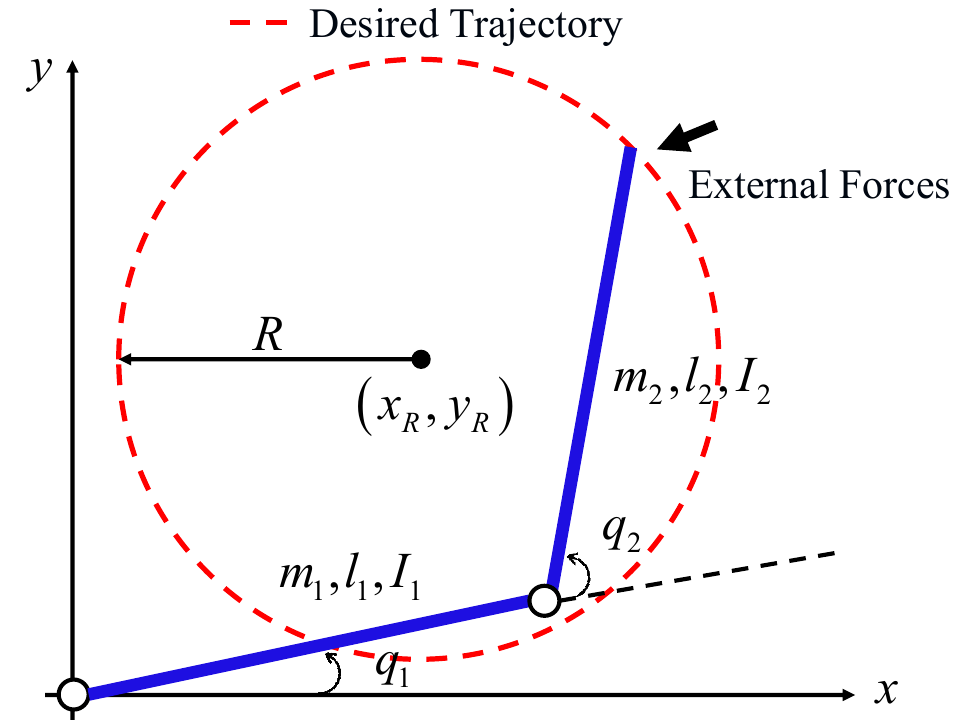}
    \caption{Simulation settings for two-link robot manipulator.}
    \label{fig1}
\end{figure}
A fully actuated two-link robot manipulator is considered on a horizontal surface, introduced in Fig. \ref{fig1}. Accordingly, the kinematics of the robot are presented in the following equation.
\begin{equation}
\mathbf{x}\left( \mathbf{q} \right)=\left[ \begin{matrix}
   {{l}_{1}}\cos \left( {{q}_{1}} \right)+{{l}_{2}}\cos \left( {{q}_{1}}+{{q}_{2}} \right)  \\
   {{l}_{1}}\sin \left( {{q}_{1}} \right)+{{l}_{2}}\sin \left( {{q}_{1}}+{{q}_{2}} \right)  \\
\end{matrix} \right],
\label{eq56}
\end{equation}
in which $\mathbf{x}\left( \mathbf{q} \right)={{\left[ \begin{matrix}
   x & y  \\
\end{matrix} \right]}^{\top}}
$ and $\mathbf{q}={{\left[ \begin{matrix}
   q_1 & q_2 \\
\end{matrix} \right]}^{\top}}
$ represent the position of the end-effector and the angles of the links, respectively. The Jacobian matrix is computed as $\mathbf{J}\left( \mathbf{q} \right)=\frac{\partial \mathbf{x}}{\partial \mathbf{q}}$. The inertia matrix and the Coriolis, centrifugal matrix in the form of \eqref{eq1} are modeled as
\begin{equation}
\begin{matrix}
   \mathbf{M}\left( \mathbf{q} \right)=\left[ \begin{matrix}
   {{M}_{1}}+{{M}_{2}}+2R\cos \left( {{q}_{2}} \right) & {{M}_{2}}+R\cos \left( {{q}_{2}} \right)  \\
   {{M}_{2}}+R\cos \left( {{q}_{2}} \right) & {{M}_{2}}  \\
\end{matrix} \right],  \\
   \mathbf{C}\left( \mathbf{q},\mathbf{\dot{q}} \right)=\left[ \begin{matrix}
   -R{{{\dot{q}}}_{2}}\sin \left( {{q}_{2}} \right) & -R\left( {{{\dot{q}}}_{1}}+{{{\dot{q}}}_{2}} \right)\sin \left( {{q}_{2}} \right)  \\
   R{{{\dot{q}}}_{1}}\sin \left( {{q}_{2}} \right) & 0  \\
\end{matrix} \right],  \\
\end{matrix}
\label{eq57}
\end{equation}
where ${{M}_{1}}=l_{1}^{2}\left( \frac{{{m}_{1}}}{4}+{{m}_{2}} \right)+{{I}_{1}}$, ${{M}_{2}}=\frac{{{m}_{2}}l_{2}^{2}}{4}+{{I}_{2}}$, and $R=\frac{{{m}_{2}}{{l}_{1}}{{l}_{2}}}{2}$. The mass, length, and inertia of the two links are accordingly installed as: ${{m}_{1}}={{m}_{2}}=3.05\text{kg}$, ${{l}_{1}}={{l}_{2}}=0.5\text{m}$, and ${{I}_{2}}={{I}_{1}}=0.0414\text{kg}{{\text{m}}^{\text{2}}}$. The mass of the fictitious flywheel is ${{m}_{f}}=10\text{kg}$. The initial conditions are set to ${{\mathbf{q}}^{a}}\left( 0 \right)={{\left[ \begin{matrix}
   1.29 & -1.67 & 0  \\
\end{matrix} \right]}^{\top}}\text{rad}
$, ${{\mathbf{\dot{q}}}^{a}}\left( 0 \right)={{\left[ \begin{matrix}
   0.5 & 0.5 & 1.3  \\
\end{matrix} \right]}^{\top}}\text{rad/s}
$, and $\mathbf{x}\left( 0 \right)={{\left[ \begin{matrix}
   0.6 & 0.3  \\
\end{matrix} \right]}^{\top}}\text{m}
$.

From Fig. \ref{fig1}, the desired trajectory is considered as the circular trajectory with the radius $R=0.3\,\text{m}$ and the center $x_R=y_R=0.35\,\text{m}$. In addition, the control parameters of the proposed method in \eqref{eq7}-\eqref{eq9} are installed as follows: $\pmb{\psi}=30\mathbf{I}_2$, ${{E}^{a}}=10\,\text{J}$, ${{\zeta}_{1}}=3$, ${{\zeta}_{2}}=5$, $k_{a}^{d}=10\text{J}$,
${{\mathbf{K}}_{1}}=2{{\mathbf{I}}_{3}}$, ${{\mathbf{K}}_{2}}=2{{\mathbf{I}}_{3}}$, $\kappa =0.5$, ${{\delta }_{1}}=\delta_4=0.01$, ${{\delta }_{2}}=\delta_3=1$, ${{\eta }_{\min }}={{\eta }_{\max }}=1$.
\subsection{Numerical Simulation under External Disturbances}
\label{4b}
%and its adaptability to control gain adjustment.
In this section, three distinct types of external disturbances from the physical environment are considered to evaluate the effectiveness of the proposed method. Specifically, 
\begin{itemize}
\item{Disturbance 1 (Artificial disturbances): Periodic external disturbances 
\begin{equation}
{{\pmb{\tau }}_{ext}}={{\left[ \begin{matrix}
   0.5\cos \left( 2t \right) & 0.5\sin \left( \frac{3}{2}t \right)  \\
\end{matrix} \right]}^{\top}},
\label{eq58}
\end{equation}
are applied to simulate oscillatory interactions with the physical environment.} 
\item{Disturbance 2 (Friction disturbances): Combined viscous and Coulomb friction is considered to mimic resistive forces from physical contact with the environment.
\begin{equation}
\pmb{\tau}_{{ext}} = 
\begin{bmatrix}
   -\text{sgn}(\dot{q}_1)\left( c_1 |\dot{q}_1| + \bar{\tau}_1 \right) \\
   -\text{sgn}(\dot{q}_2)\left( c_2 |\dot{q}_2| + \bar{\tau}_2 \right)
\end{bmatrix},
\label{eq59}
\end{equation}
with ${{c}_{1}}={{c}_{2}}=0.1,\,{{\bar{\tau }}_{1}}={{\bar{\tau }}_{2}}=0.5$.}
\item{Disturbance 3 (Pushing disturbances): The disturbances have the same direction as the robot motion.
\begin{equation}
{{\pmb{\tau }}_{ext}}={{\left[ \begin{matrix}
   {{{\bar{\tau }}}_{3}}\frac{{{{\dot{q}}}_{1}}}{\left| {{{\dot{q}}}_{1}} \right|} & {{{\bar{\tau }}}_{4}}\frac{{{{\dot{q}}}_{2}}}{\left| {{{\dot{q}}}_{2}} \right|}  \\
\end{matrix} \right]}^{\top}},
\label{eq60}
\end{equation}
with ${{\bar{\tau }}_{3}}={{\bar{\tau }}_{3}}=0.9$.}
\end{itemize}

\begin{figure*}
    \centering
    \includegraphics[width=0.9\linewidth]{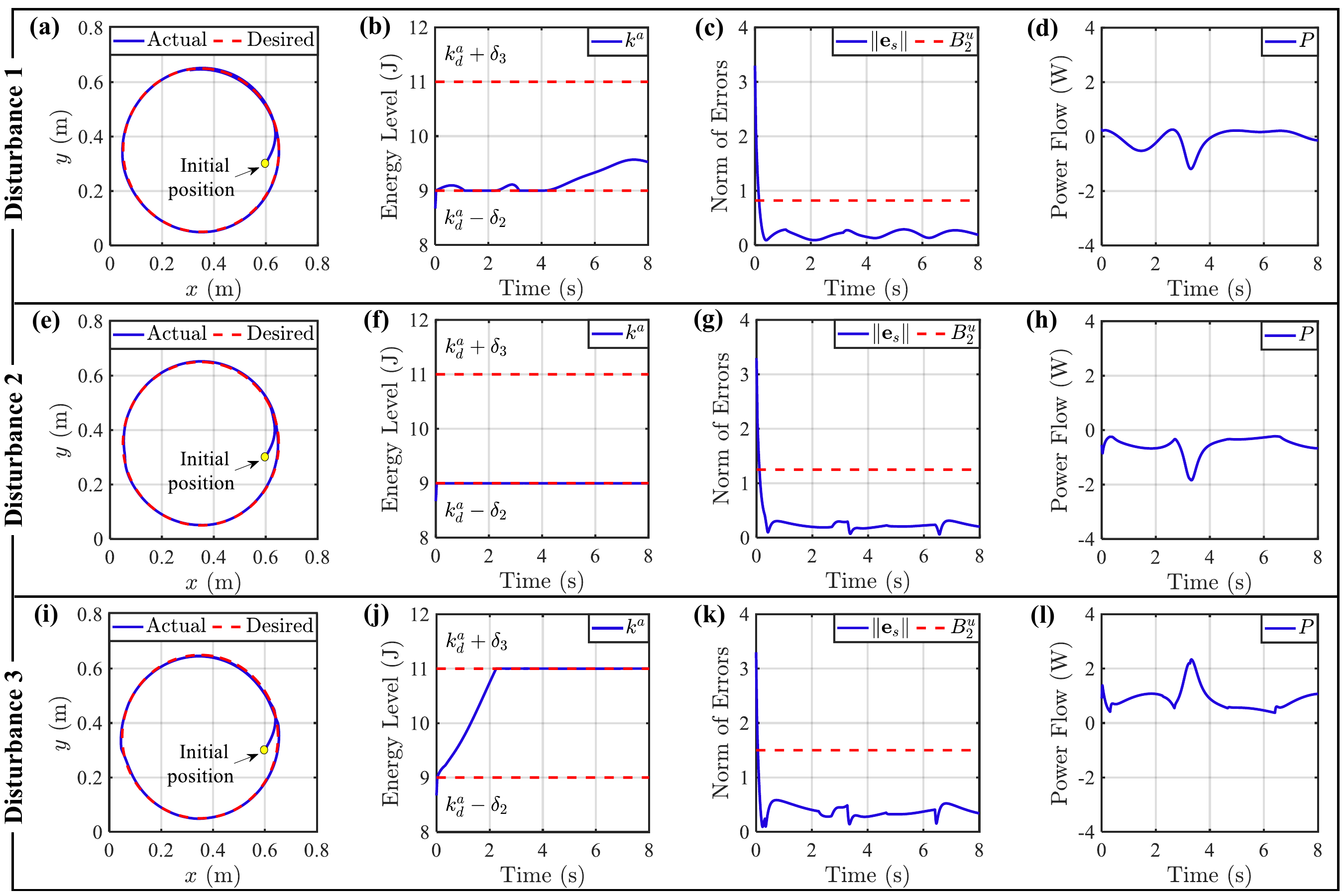}
    \caption{Responses of the closed-loop system in Section \ref{4b}, where (a), (e), (i) are actual and desired trajectories; (b), (f), (j) are energy levels; (c), (g), (k) are norms of tracking errors; (d), (h), (l) are power flows under Disturbances 1, 2, and 3, respectively.}
    \label{f2}
\end{figure*}

Fig. \ref{f2} shows the simulation results under three different external disturbances in \eqref{eq58}, \eqref{eq59}, and \eqref{eq60}. Specifically, the trajectory tracking performance of the robot manipulator under Disturbances 1, 2, and 3 are illustrated in Figs. \ref{f2}(a), (e), (i), respectively. Although the robot system is affected by external disturbances from the physical environment, it demonstrates robust capability in tracking the desired trajectory. In Figs. \ref{f2}(b), (f), (j), the energy levels of the closed-loop system converge to the bounded region $\left[ k_{d}^{a}+{{\delta }_{2}},k_{d}^{a}+{{\delta }_{3}} \right]=\left[ 9,11 \right]\text{J}$ within a finite-time interval (approximately 0.1s) in all external disturbance scenarios, as proven in Theorem \ref{thm3}. For Disturbance 1, the energy level in Fig. \ref{f2}(b) exhibits variation within the bounded domain due to the artificial disturbances in \eqref{eq58} constructed by sine and cosine functions. Fig. \ref{f2}(f) under Disturbance 2 maintains the energy level at the lower bound because of frictional forces \eqref{eq59} acting in the direction opposite to the robot's motion, whereas Fig. \ref{f2}(j) under Disturbance 3 tends to approach the upper bound as a result of external disturbances \eqref{eq60} that reinforce the movement. Additionally, the results in Figs. 2(c), (g), and (k) demonstrate the correctness as stated in Theorem 4. The norms of tracking errors, including position and velocity errors, converge to the bounded regions $B_2\le B_{2}^{u}$ of Disturbances 1, 2, and 3 within finite-time intervals and remain in these domains. Due to the different magnitudes of external disturbances, the upper bounds and the settling time of Disturbances 1, 2, and 3 are $ B_{2}^{u}=0.82$, $1.25$, $1.50$, and $T_s=0.24\text{s}$, $0.14\text{s}$, $0.05\text{s}$, respectively. As indicated by \eqref{eq51} in Proof of Theorem \ref{thm4}, the settling time $T_s$ is inversely related to the error bound $B_{2}^{u}$, which accounts for the observed convergence rates in Disturbances 1, 2, and 3. Furthermore, the power flows between the closed-loop robotic system and its physical environment for Disturbances 1, 2, and 3 are also shown in Figs. \ref{f2}(d), (h), (l). From theoretical viewpoints, the power flows in all external disturbance scenarios satisfy the bounded region in Corollary \ref{coro1}, which is derived from the results of Theorems \ref{thm2}, \ref{thm3}, and \ref{thm4}. Due to the influence of external disturbances, the power flow under Disturbance 1 has directional variations, whereas under Disturbances 2 and 3, unidirectional power flows are observed, from the robotic system to the environment in Disturbance 2 $\left(P\left( t \right)<0\right)$, and from the environment to the robotic system in Disturbance 3 $\left(P\left( t \right)>0\right)$.

%In each example, we consider three scenarios as:
%begin{itemize}
   % \item{The first scenario: The external disturbances in \eqref{eq58}, \eqref{eq59}, \eqref{eq60}, and the sum of them are applied directly.} 
   % \item{The second scenario: The higher-magnitude disturbances are considered by doubling the disturbance magnitudes of the first scenario in each example.} 
   % \item{The third scenario:} 
%\end{itemize}
%As indicated by \eqref{eq51} in Proof of Theorem \ref{thm4}, the settling time decreases with increasing control parameters and disturbance magnitudes, which accounts for the observed convergence rates in Scenarios 1.1, 1.2, and 1.3. 

\subsection{Numerical Simulation of Control Parameter Effect}
\label{4c}
To investigate the effect of the control parameter, the following conditions are first considered under Disturbance 1 in \eqref{eq58}.  
\begin{itemize}
\item{Condition 1: Nominal control parameters as defined in Section \ref{4a} under the artificial disturbances \eqref{eq58}.} 
\item{Condition 2: The control parameter $\kappa=1.0$ rather than $\kappa=0.5$ in Section \ref{4a} under the disturbances \eqref{eq58}.}
\item{Condition 3: The control parameter $\kappa = 1.0$ under disturbances of 1.5-times amplified amplitude in \eqref{eq58}.}
\end{itemize}
The increase in the control parameter $\kappa$ is motivated by two primary considerations. First, according to the results in Section \ref{4b}, the energy levels are bounded within the operational region under all scenarios, thereby indicating that tuning should prioritize parameters related to tracking performance. Second, as demonstrated in the proof of Theorem \ref{thm4}, the parameter $\kappa$ plays a vital role in improving tracking accuracy, as it directly influences both the error bound and the settling time of the system.
\begin{figure*}
    \centering
    \includegraphics[width=0.9\linewidth]{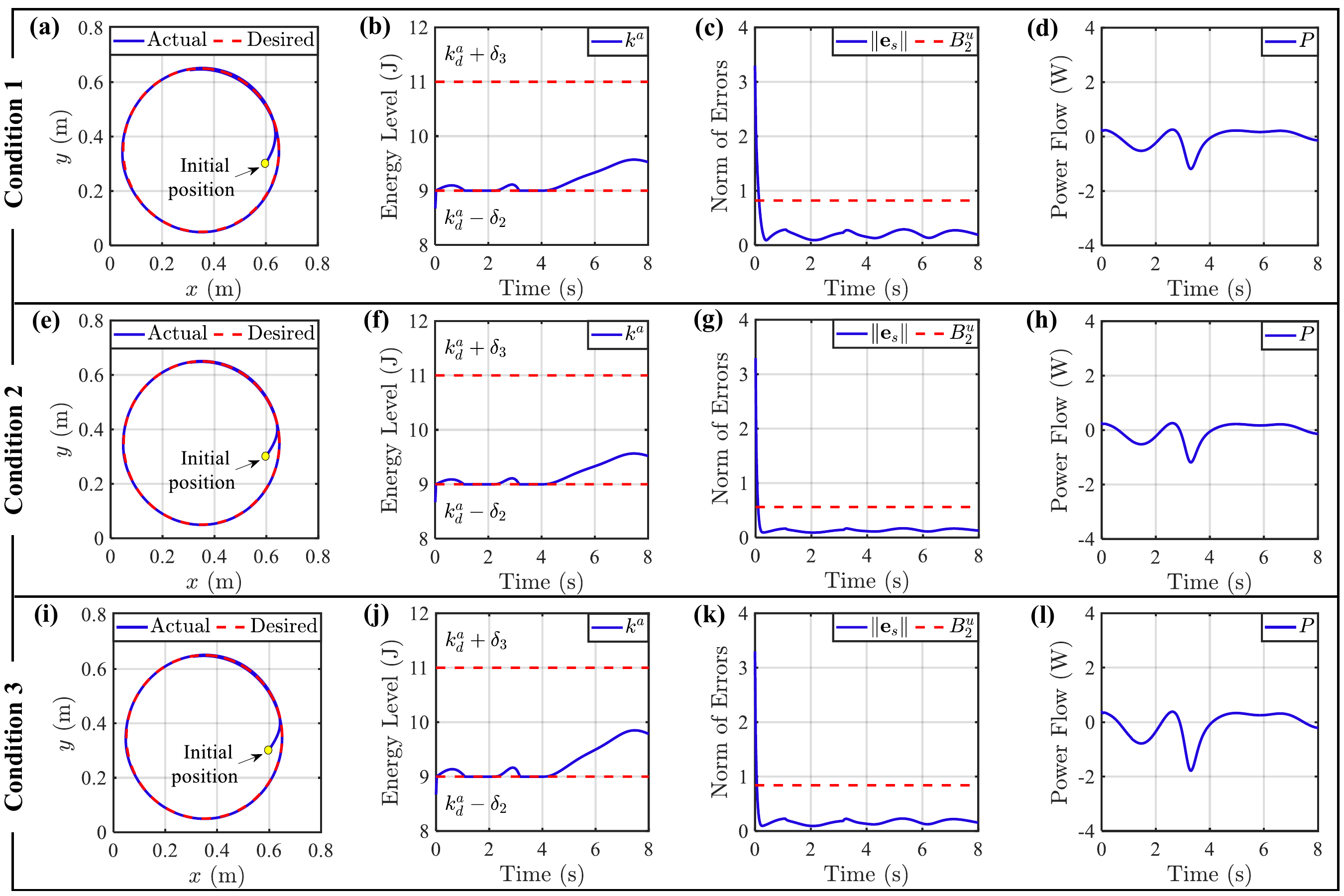}
    \caption{Responses of the closed-loop system under Disturbance 1 with control parameter adjustment, where (a), (e), (i) are actual and desired trajectories; (b), (f), (j) are energy levels; (c), (g), (k) are norms of tracking errors; (d), (h), (l) are power flows of Conditions 1, 2, and 3, respectively.}
    \label{f3}
\end{figure*}

The simulation results of Conditions 1, 2, and 3 under Disturbance 1 are presented in Fig. \ref{f3}. In comparison with the results in Figs. \ref{f3}(a), (c) of Condition 1, the tracking performance in Figs. \ref{f3} (e), (g), and (i), (k) of Conditions 2 and 3 are significantly enhanced even under the higher-magnitude disturbances in Condition 3. Specifically, the norm of tracking errors under Condition 2 converges to a bounded region $B_2\le B_{2}^{u}=0.56$, which is approximately 1.5-fold as tight as the region $B_2\le B_{2}^{u}=0.82$ in Condition 1. For Condition 3, despite the presence of intensified disturbances, the norm of tracking errors is bounded within $B_2 \le B_2^{u} = 0.84$ after a finite-time interval, offering a similar bound as in Condition 1. Furthermore, Conditions 2 and 3 exhibit faster convergence, with settling times of $0.12\text{s}$ and $0.05\text{s}$, respectively, rather than $0.24\text{s}$ observed in Condition 1. Accordingly, the application of the increased control parameter in Conditions 2 and 3 leads to notable reductions in tracking errors, bounded region sizes, and convergence times, consistent with the properties and analytical insights established in Theorem \ref{thm4} and Remark \ref{remark3}. Meanwhile, the energy levels across all three conditions in Figs. \ref{f3}(b), (f), (j) consistently satisfy the prescribed bounded working region $\left[ k_{d}^{a}+{{\delta }_{2}},k_{d}^{a}+{{\delta }_{3}} \right]=\left[ 9,11 \right]\text{J}$, regardless of variations in the control parameter $\kappa$ and disturbance magnitude, as theoretically established in Theorems \ref{thm2}, \ref{thm3} and further corroborated by the analyses in Remarks \ref{remark2}, \ref{remark4}. Moreover, the power flow in Fig. \ref{f3}(l) of Condition 3 is roughly 1.5 times greater than that in Figs. \ref{f3}(d), (h) of Conditions 1, 2, primarily owing to the proportional increase in disturbance magnitudes and the linear nature of the power flow formulation in \eqref{eq52}. The power flows in Fig. \ref{f3}(d) of Condition 1 and Fig. \ref{f3}(h) of Condition 2 are nearly identical because they are subjected to the same external disturbances and exhibit only minor differences in velocity tracking errors.

Note that the simulation results of Conditions 1, 2, and 3 under Disturbances 2 and 3 are also presented in Figs. \ref{f6} and \ref{f7} of Appendix. It is necessary to state that the aforementioned evaluations for Conditions 1, 2, and 3 under Disturbance 1 are also observed in Figs. \ref{f6} and \ref{f7}.

\subsection{Comparative Simulations}
\label{4d}

\begin{figure*}
    \centering
    \includegraphics[width=0.9\linewidth]{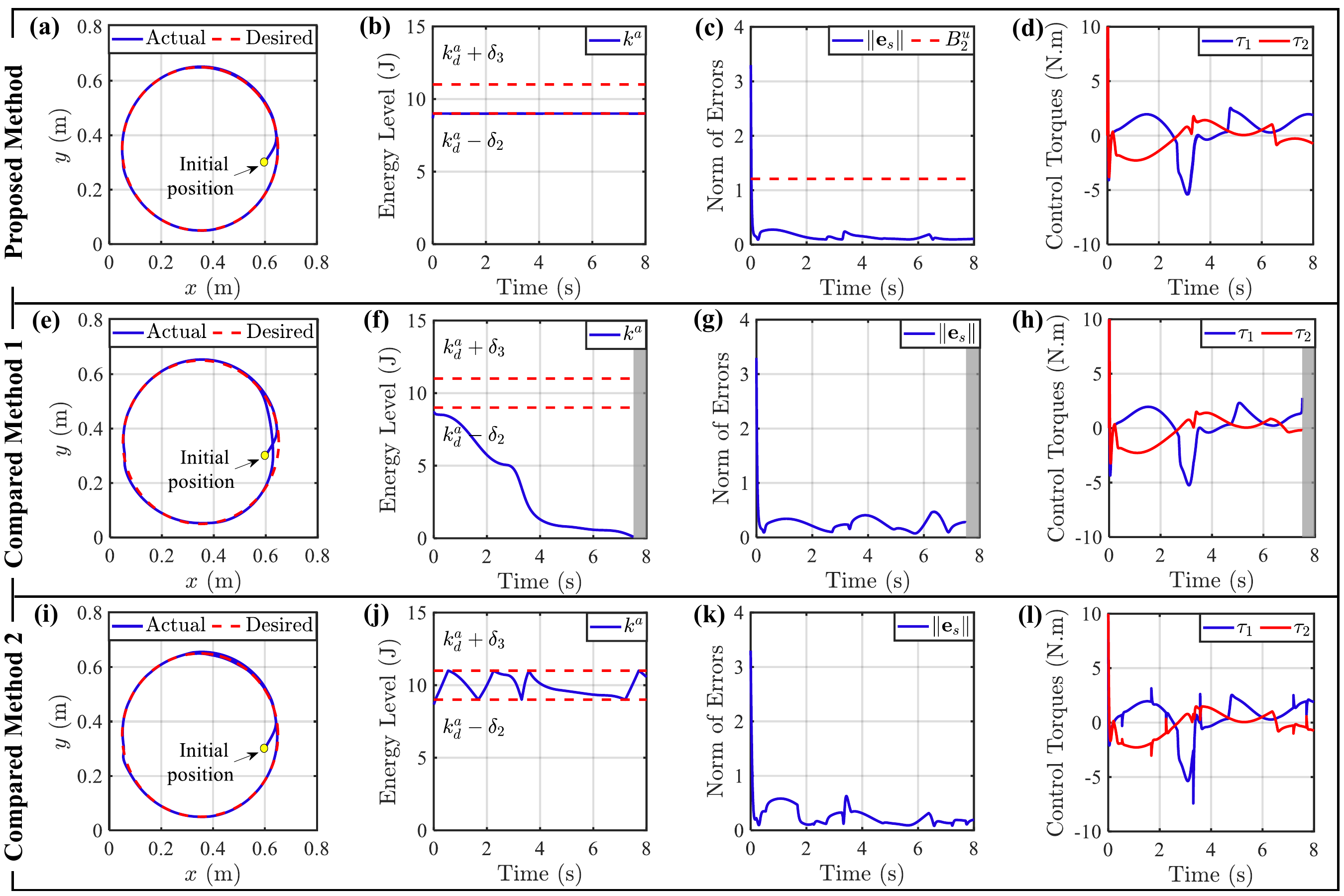}
    \caption{Comparative results under the combined effect of doubled Disturbances 1 and 2, where (a), (e), (i) are actual and desired trajectories; (b), (f), (j) are energy levels; (c), (g), (k) are norms of tracking errors; (d), (h), (l) are robot control torques of three methods.}
    \label{f4}
\end{figure*}
\begin{figure*}
    \centering
    \includegraphics[width=0.9\linewidth]{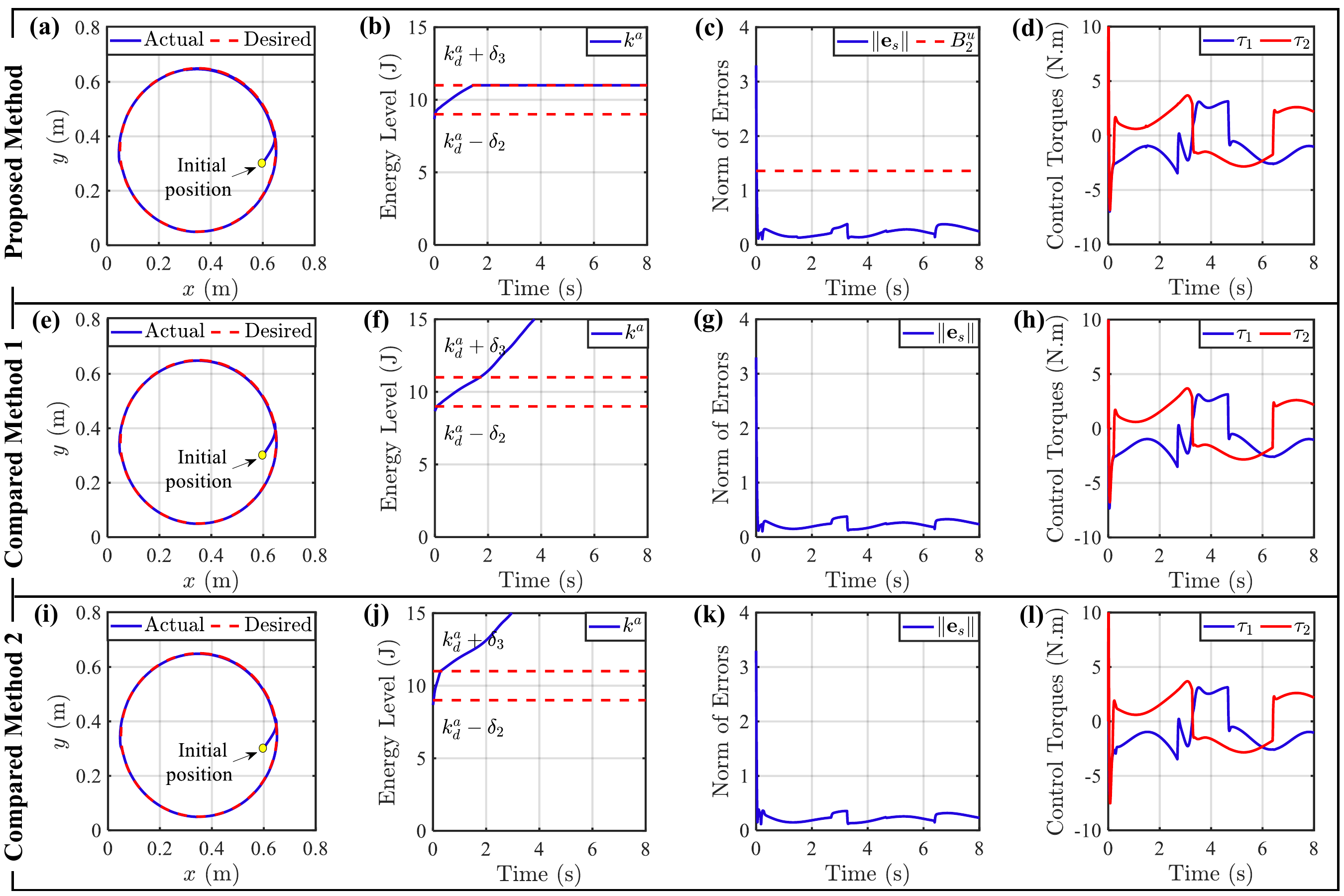}
    \caption{Comparative results under the combined effect of doubled Disturbances 1 and 3, where (a), (e), (i) are actual and desired trajectories; (b), (f), (j) are energy levels; (c), (g), (k) are norms of tracking errors; (d), (h), (l) are robot control torques of three methods.}
    \label{f5}
\end{figure*}

To highlight the effectiveness of the proposed method in guaranteeing task performance and safe interaction, comparative numerical simulations are implemented in this section. Specifically,
\begin{itemize}
\item{Proposed Method: The continuous SPVFC in \eqref{eq8}.} 
\item{Compared Method 1: The original PVFC in \cite{c22} with the time-varying velocity field \eqref{eq7}.}
\item{Compared Method 2: The semi-passive control framework in \cite{c28}, switching between conservative and nominal controllers. To ensure a fair comparison, the conservative controller is selected as the original PVFC \cite{c22} to preserve energetic passivity, whereas the nominal controller, defined as ${{\pmb{\tau }}^{a}}\left( {{\mathbf{q}}^{a}},{{{\mathbf{\dot{q}}}}^{a}},t \right)={{\mathbf{R}}_{1}}\left( {{\mathbf{q}}^{a}},{{{\mathbf{\dot{q}}}}^{a}},t \right){{{\mathbf{\dot{q}}}}^{a}}+{{\mathbf{R}}_{2}}\left( {{\mathbf{q}}^{a}},{{{\mathbf{\dot{q}}}}^{a}},t \right){{{\mathbf{\dot{q}}}}^{a}}+\mathbf{K}{{{\mathbf{\dot{q}}}}^{a}}$ with the injected energy term $\mathbf{K}{{{\mathbf{\dot{q}}}}^{a}}$, is adopted to achieve the nominal performance. The method switches to the nominal mode when the energy falls below $k_{d}^a-\delta_2$, and switches to the conservative mode when it exceeds $k_{d}^a+\delta_3$; otherwise, the current control mode is maintained.}
\end{itemize}
It should be noted that all three methods use the parameters in Section \ref{4a}, $\kappa=2.0$, and $\mathbf{K}=\mathbf{K}_1=\mathbf{K}_2=2{{\mathbf{I}}_{3}}$.

Fig. \ref{f4} presents comparative results of the three methods under the combined disturbance, where the magnitudes of Disturbances 1 \eqref{eq58} and 2 \eqref{eq59} are doubled. Specifically, the effectiveness of the Proposed Method further demonstrates the statements in Section \ref{4b}, \ref{4c} about robustness, energy constraint, and tracking error boundedness via the results in Figs. \ref{f4}(a), (b), (c). Meanwhile, the Compared Method 1 exhibits an energy-depletion phenomenon at $t\approx7.5 \text{s}$ in Fig. \ref{f4}(f) due to the external disturbances and the absence of an energy-compensation mechanism. Consequently, its task performance is significantly degraded, as illustrated in Figs. \ref{f4}(e), (g), and the robotics system is terminated at $t\approx7.5 \text{s}$. The gray-shaded regions in Figs. \ref{f4}(f), (g), and (h) of the Compared Method 1 denote the periods during which the robot operation is halted. Although the Compared Method 2 is capable of compensating for the energy level, as shown in Fig. \ref{f4}(j), its discontinuous mode-transition behavior induces transient peaking phenomena that degrade the overall tracking performance, as observed in Figs. \ref{f4}(i) and (k) and discussed in Remark 7 of the original work \cite{c28}. Furthermore, the robot control inputs corresponding to the three methods are presented in Figs. \ref{f4}(d), (h), and (l), where the control torques in Fig. \ref{f4}(h) further illustrate the peaking behavior in the Compared Method 2.

Fig. \ref{f5} shows comparative results under combined disturbance conditions with doubled Disturbances 1 \eqref{eq58} and 3 \eqref{eq60}, which tend to inject energy into the closed-loop system. In this context, there is no energy-depletion issue in the Compared Method 1, and both the Compared Methods 1 and 2 achieve task performance as the Proposed Method, as presented in Figs. \ref{f5}(a), (c), (e), (g), (i), (k). While the Proposed Method preserves the prescribed energy bounds in Fig. \ref{f5}(b), the Compared Methods 1 and 2 accumulate all energy injected from the physical environment, thereby violating the upper energy bound in Figs. \ref{f5}(f) and (j). This implies that excessive energy accumulation in the Compared Methods 1 and 2 may lead to the sudden release of a large amount of stored energy, thereby posing potential safety and stability risks. Moreover, the robot control torques of the three methods are given in Figs. \ref{f5} (d), (h), and (l). It should be noted that almost no peaking phenomenon caused by the switching behavior is observed in the Compared Method 2, since the conservative controller is primarily maintained in this situation.

\section{Discussion and Limitation}
In the present study, the robust time-varying semi-passive velocity field control method for the class of nonlinear mechanical systems subjected to external disturbances was developed, analyzed, and validated. From theoretical viewpoints, the proposed method enables relaxing the conservative nature of conventional PVFC in a controlled manner based on the information of the energy level, as mentioned in Section \ref{3b}. According to Theorems \ref{thm3} and \ref{thm4}, the proposed control strategy simultaneously ensured the convergence of both the energy level and the tracking errors of the closed-loop system to their respective bounded regions under external disturbances, addressing a limitation unconsidered in the preliminary work \cite{c27}. Based on the theoretical analyses from Theorems \ref{thm2}, \ref{thm3}, and \ref{thm4}, systematic remarks were given to emphasize how the control parameters could be effectively configured for practical implementation, as summarized in Remarks \ref{remark1}, \ref{remark2}, \ref{remark3}, and \ref{remark4}. Furthermore, the correctness of the proposed method was also validated through simulation examples of a two-link robot manipulator subjected to multiple external disturbances. Beyond ensuring robustness and UUB, the simulation results in Section \ref{4c} further demonstrated the effectiveness of tuning control parameters in specific conditions. Moreover, the comparative results in Section \ref{4d} highlighted the capability of the proposed method to relax the conservative nature, eliminate the discontinuous behavior, and mitigate the potential risks in \cite{c22} and \cite{c28}. Additionally, although the power flows in all scenarios satisfy the bounded region in Corollary \ref{coro1} from theoretical viewpoints, the bounded domain of the power flow remains relatively broad due to using the results of Theorems \ref{thm2}, \ref{thm3}, and \ref{thm4} directly. Therefore, it is necessary to incorporate power flow adjustment mechanisms into the current framework in future work to achieve a tighter bound.

\section{Conclusion}

In this paper, we have introduced a robust time-varying semi-passive velocity field control approach to address the inherent conservativeness of classical PBC and enhance physical interaction safety. The proposed method enabled relaxing conservative nature in a controlled manner by ensuring that the system behaved passively with respect to the pair of external force and velocity when the internal energy exceeded a designed threshold, while allowing non-passive behavior when necessary for task execution. Theoretical results have demonstrated the convergence of both the energy level and system states toward the bounded regions despite the existence of external disturbances. Furthermore, a constraint on instantaneous power flow was also guaranteed by the proposed approach. The effectiveness and advantages of the proposed approach were validated by numerical simulation examples. 

For future work, the proposed method will be further investigated under input saturation to improve its suitability for practical implementations. Moreover, a power valve-like mechanism will be developed to directly regulate the rate of energy delivered between the closed-loop mechanical system and its physical environment. Building upon this foundation, future research will also explore the application of the proposed approach to physical interaction scenarios, such as human--robot interaction tasks, and investigate its integration with preference-based optimization to enable personalized or task-specific adaptation of control parameters.
% A conclusion section is not required. Although a conclusion may review the 
% main points of the paper, do not replicate the abstract as the conclusion. A 
% conclusion might elaborate on the importance of the work or suggest 
% applications and extensions. 

%\appendices
\section*{Appendix}
\appendices
\label{app}
Figs. \ref{f6} and \ref{f7} show the results of Conditions 1, 2, and 3 under Disturbances 2 and 3, respectively.
\renewcommand{\thefigure}{A.\arabic{figure}}
\setcounter{figure}{0}

\begin{figure*}
    \centering
    \includegraphics[width=0.9\linewidth]{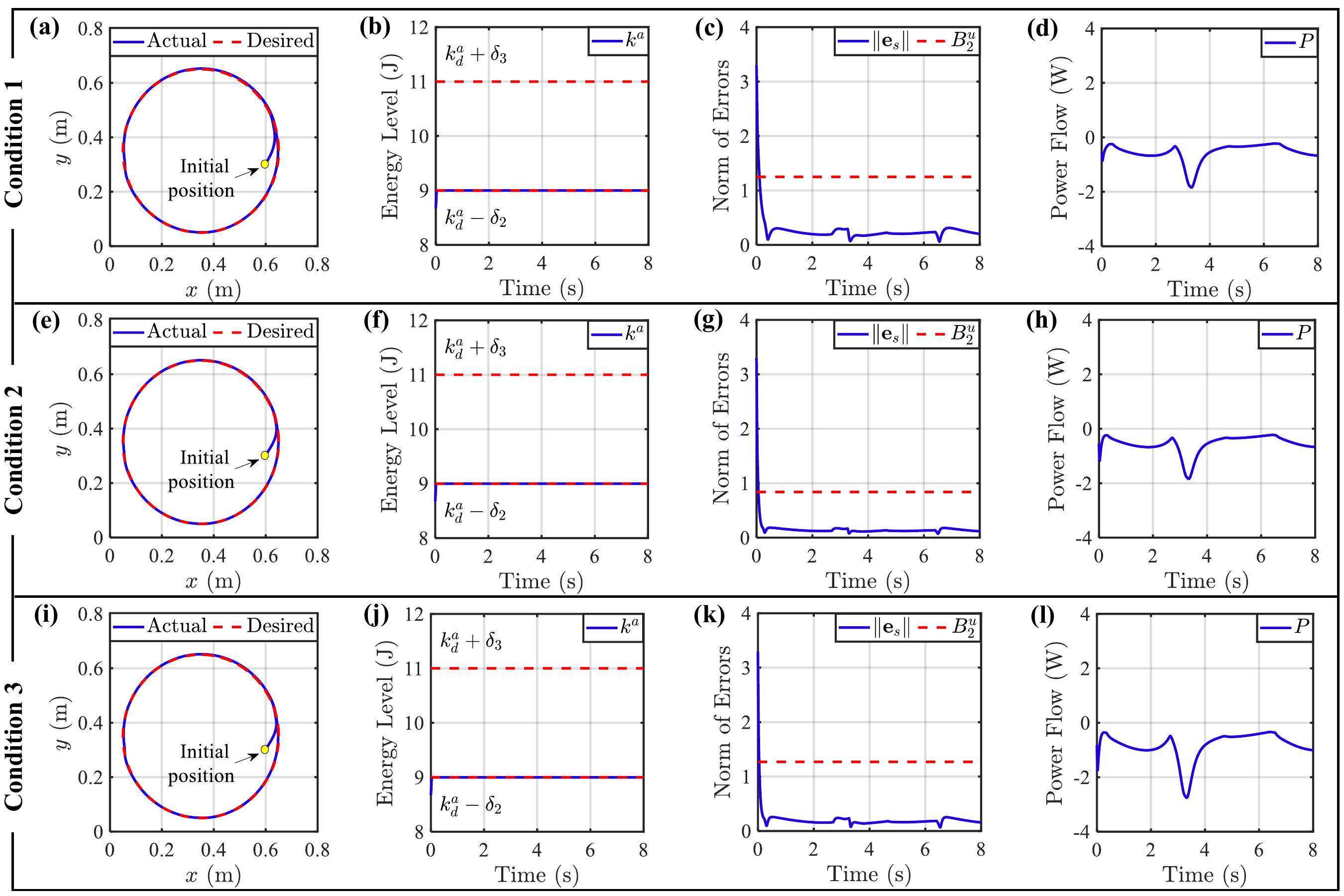}
    \caption{Responses of the closed-loop system under Disturbance 2 with control parameter adjustment, where (a), (e), (i) are actual and desired trajectories; (b), (f), (j) are energy levels; (c), (g), (k) are norms of tracking errors; (d), (h), (l) are power flows of Conditions 1, 2, and 3, respectively.}
    \label{f6}
\end{figure*}

\begin{figure*}
    \centering
    \includegraphics[width=0.9\linewidth]{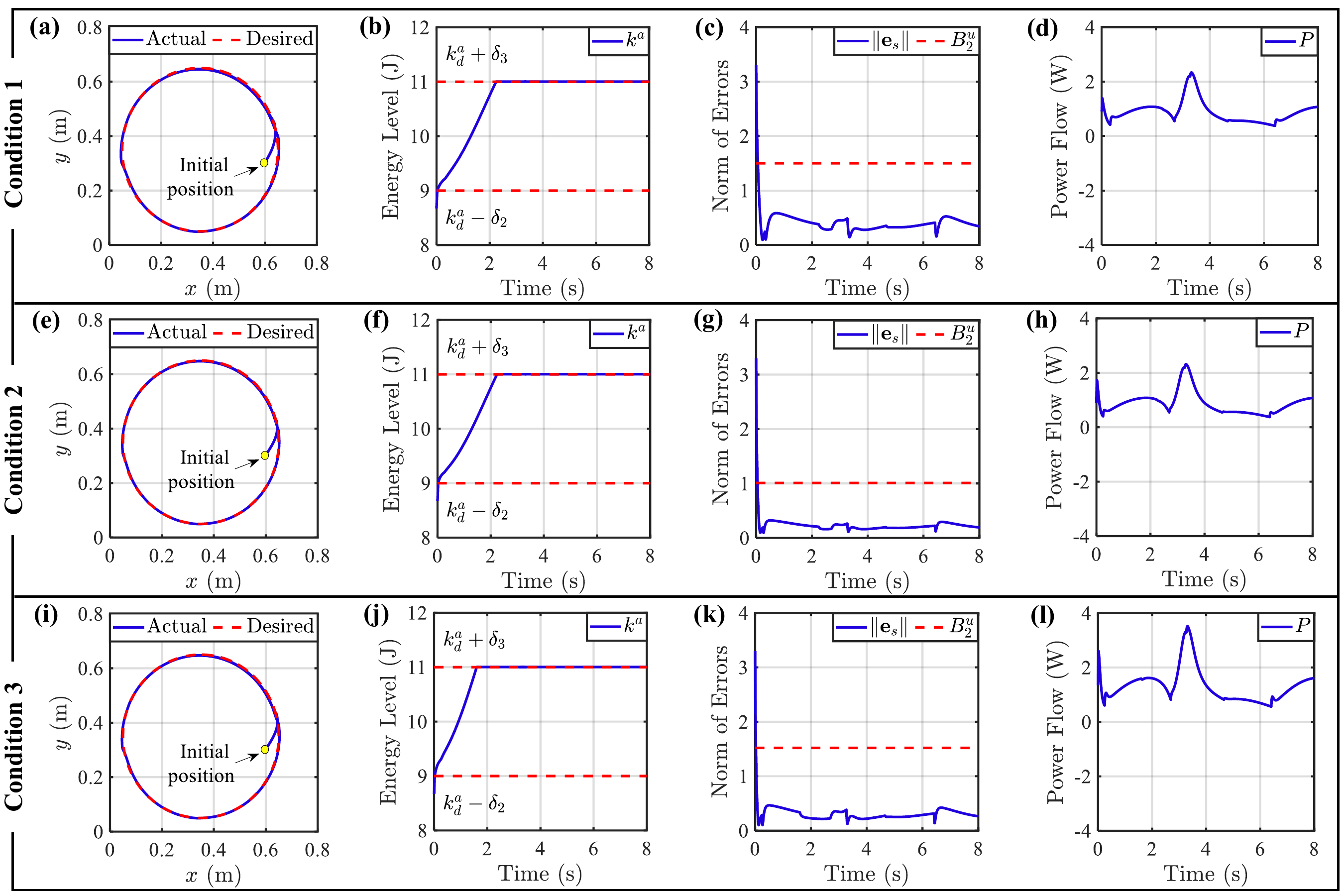}
    \caption{Responses of the closed-loop system under Disturbance 3 with control parameter adjustment, where (a), (e), (i) are actual and desired trajectories; (b), (f), (j) are energy levels; (c), (g), (k) are norms of tracking errors; (d), (h), (l) are power flows of Conditions 1, 2, and 3, respectively.}
    \label{f7}
\end{figure*}

\vspace{-3.75em}
\begin{IEEEbiography}[{\includegraphics[width=1in,height=1.25in,clip,keepaspectratio]{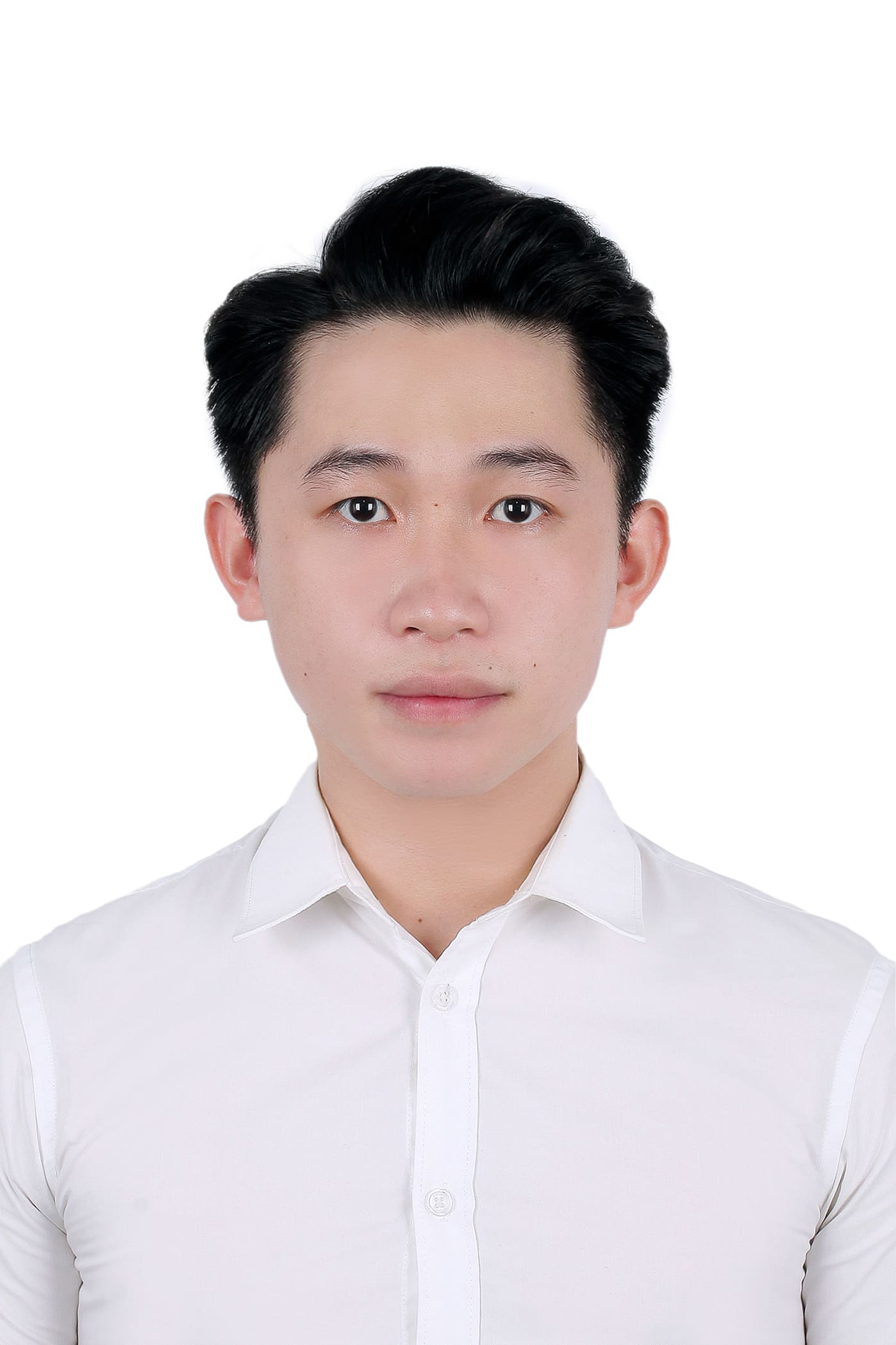}}]{Van Trong Dang} (Graduate Student Member, IEEE) received the B.E. and M.Sc. degrees in control engineering and automation from the School of Electrical and Electronic Engineering, Hanoi University of Science and Technology, Hanoi, Vietnam, in 2021 and 2022, respectively. 

 Since 2023, he has been a Ph.D. student with the Division of Information Science, Nara Institute of Science and Technology, Nara, Japan. His research interests include control theory and its applications, with special emphasis on stabilization, passivity-based control, energy-based control, safety-critical control, and collaborative human-robot systems.

\end{IEEEbiography}
\vspace{-4.5em}
\begin{IEEEbiography}[{\includegraphics[width=1in,height=1.25in,clip,keepaspectratio]{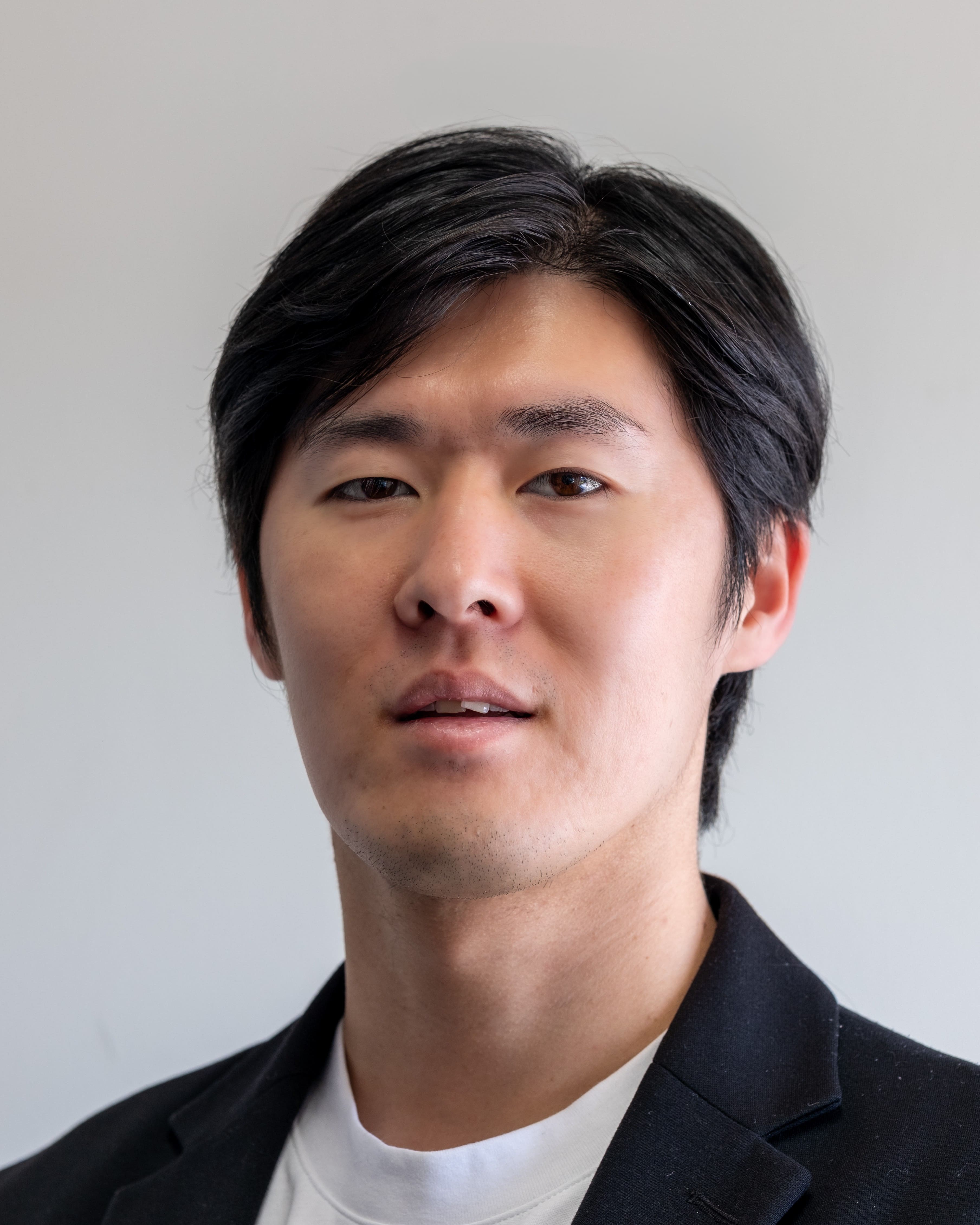}}]{Sumitaka Honji} (Member, IEEE) received the B.E., M.E., and Ph.D. degrees in mechanical engineering from Kyushu University, Fukuoka, Japan, in 2019, 2021, and 2024, respectively. 

 Since 2024, he has been an assistant professor with the Graduate School of Science and Technology, Nara Institute of Science and Technology, Ikoma, Japan. His research interests include modeling and control of soft robots, control theory, and human-robot interaction.

 Dr. Honji is a member of the Robotics Society of Japan
and the Japan Society of Mechanical Engineers.

\end{IEEEbiography}
\vspace{-4.5em}
\begin{IEEEbiography}[{\includegraphics[width=1in,height=1.25in,clip,keepaspectratio]{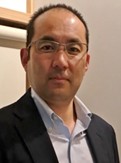}}]{Takahiro Wada} (Member, IEEE) received the B.S. degree in mechanical engineering, the M.S. degree in information science and systems engineering, and the Ph.D. degree in robotics from Ritsumeikan University, Japan, in 1994, 1996, and 1999, respectively.  

 In 1999, he became an Assistant Professor with Ritsumeikan University. In 2000, he joined Kagawa University, Takamatsu, Japan, as an Assistant Professor in the Department of Intelligent Mechanical Systems Engineering as a Faculty of engineering. He was promoted to associate professor in 2003. In 2012, he joined the College of Information Science and Engineering at Ritsumeikan University as a Full Professor. Since 2021, he has been a Full Professor at the Nara Institute of Science and Technology (NAIST), Japan. In 2006 and 2007, he spent half a year at the University of Michigan Transportation Research Institute in Ann Arbor as a Visiting Researcher. His current research interests include human–machine systems, robotics, and human modeling. 

\end{IEEEbiography}

 %\begin{IEEEbiographynophoto}{Second B. Author,} photograph and biography not available at the
% time of publication.
 %\end{IEEEbiographynophoto}

% \begin{IEEEbiographynophoto}{Third C. Author Jr.} (Member, IEEE), photograph and biography not available at the
% time of publication.
% \end{IEEEbiographynophoto}

\end{document}